\documentclass[pdflatex,sn-basic]{sn-jnl}

\usepackage{graphicx}%
\usepackage{multirow}%
\usepackage{amsmath,amssymb,amsfonts}%
\usepackage{amsthm}%
\usepackage{mathrsfs}%
\usepackage[title]{appendix}%
\usepackage{xcolor}%
\usepackage{textcomp}%
\usepackage{manyfoot}%
\usepackage{booktabs}%
\usepackage{algorithm}%
\usepackage{algorithmicx}%
\usepackage{algpseudocode}%
\usepackage{listings}%
\usepackage{url}

\theoremstyle{thmstyleone}%

\theoremstyle{thmstyletwo}%

\theoremstyle{thmstylethree}%
\usepackage{mathtools}
\usepackage{bm}
\usepackage{placeins}
\numberwithin{equation}{section}

\theoremstyle{thmstyleone}
\newtheorem{lemma}{Lemma}
\newtheorem{thm}{Theorem}[section]
\newtheorem{cor}{Corollary}[section]
\theoremstyle{thmstyletwo}
\newtheorem{rem}{Remark}[section]

\newcounter{FSDmainfigures}
\newcounter{FSDmaintables}
\newcommand{\FSDAuthorContributions}{YO and HN conceived the study. YO developed the mathematical arguments and conducted the formal analysis, and wrote the original draft. YO and HN reviewed and revised the manuscript. All authors approved the final manuscript.}

\newcommand{\FSDFundingStatement}{
This work was supported by JSPS KAKENHI (grant number 25K20599 to YO), the Japan Agency for Medical Research and Development (grant number JP26fk0108742 to YO), the SECOM Science and Technology Foundation (YO), Health and Labour Sciences Research Grants (grant number 23HA2005 to HN), HU-RIZONT international research excellence program (Rapid-GRIP project: 2024-1.2.3-HU-RIZONT-2024-00034 to HN), the World Health Organization (HN), the Japan Science and Technology Agency  CREST program (grant number JPMJCR24Q3 to HN). The funders had no role in the study design, data collection and analysis, decision to publish, or preparation of the manuscript.
}

\urldef{\FSDGitHubURL}\url{https://github.com/pk2393/FSD_calc_GWP}
\urldef{\FSDZenodoURL}\url{https://doi.org/10.5281/zenodo.20446753
        
        }
\begin{document}

\title[Total progeny size of Galton--Watson Process]{Computing total progeny distributions of multi-type Galton--Watson processes}

\author*[1]{\fnm{Yuta} \sur{Okada}}\email{okada.yuta.4y@kyoto-u.ac.jp}
\author[1,2,3]{\fnm{Hiroshi} \sur{Nishiura}}\email{hnishiur@uwo.ca}

\affil*[1]{\orgdiv{School of Public Health}, \orgname{Kyoto University}, \orgaddress{\street{Yoshida-Konoe-Cho}, \city{Sakyo Ward}, \postcode{6068501}, \state{Kyoto}, \country{Japan}}}

\affil[2]{\orgdiv{Center for Health Security}, \orgname{Kyoto University}, \orgaddress{\street{Yoshida-Konoe-Cho}, \city{Sakyo Ward}, \postcode{6068501}, \state{Kyoto}, \country{Japan}}}

\affil[3]{\orgdiv{Department of Epidemiology and Biostatistics, Schulich School of Medicine}, \orgname{Western University}, \orgaddress{\street{1465 Richmond Street}, \city{London}, \postcode{N6G 2M1}, \state{Ontario}, \country{Canada}}}

\abstract{
Type-specific total progeny counts have valuable information even when the entire genealogies of the population or events of interest are not fully observed. To facilitate the application of multi-type Galton--Watson process (GWP) in this context, we provide a probabilistic insight regarding its probability generating function (PGF). First, geometric error bounds for computing values of the implicit total progeny PGF inside the unit polydisc are obtained as a function of the radii of the polydisc. These bounds hold regardless of criticality or irreducibility of the underlying process, and also apply to a given sequence of varying environments. 
Next, we connect the PGF error bounds to errors in numerical approximation by the Cauchy coefficient formula, and derive sharper contraction bound, the smoothness of total progeny probabilities across the criticality threshold, and conditioning on extinction. 
Finally, we provide numerical examples including an epidemiological application to datasets with incomplete type information. The results provide a basis for concise calculation of the total progeny distribution, that has wide application across various fields including biology and epidemiology. 
}

\keywords{Multi-type branching process; Galton--Watson process; final epidemic size; total progeny; Cauchy integral}

\pacs[MSC Classification]{Primary 60J80; Secondary 92D30;60J85}

\maketitle

\bmhead{Acknowledgements}
We wish to thank Stuart Jenkinson, PhD, from Edanz (https://jp.edanz.com/ac) for editing a draft of this manuscript.

\bmhead{Funding}
\FSDFundingStatement

\bmhead{Author contributions}
\FSDAuthorContributions

\clearpage

\section{Introduction}
\label{sec:intro}

Stochastic processes have been studied both as a mathematical description of the real-world phenomena including epidemics and biological processes involving population dynamics.\citep{Britton2019-xn, Kimmel2015-xd, Haccou2011-aa} In addition to the real-time dynamics, one of the key focuses regarding stochastic processes is their asymptotics; in an epidemiological context, numerous studies have explored the theory surrounding the basic reproduction number and asymptotic state of epidemics as a stochastic process in a closed population with a social or network structure 
\citep{Ball1986TotSize, Lefevre1990ReedFrost, Scalia-Tomba1990AsymptoticFS_hetero, Neal2006_rf_graph, Ball2008_network_mix, Allard2009MultitypePercol, Kenah2007NetworkEpidemicsSecondLook}.

The Galton--Watson process (GWP) is among the most widely used stochastic processes in this context; it has been applied not only for describing and capturing the heterogeneous nature of transmission of certain infectious diseases, but also for describing the accumulation of genetic mutations at the cellular level.\citep{Lloyd-Smith2005SSevent, Bozic2010-te} Even when the offspring of each generation is not fully observed, the total progeny size distribution has proved to be useful for understanding the key characteristics of the offspring. In fact, total progeny distribution of single-type GWPs has been heavily applied in epidemiological studies \citep{Gay2004-ib, blumberg2013ChainSizeCompare,blumberg2013Stuttering,Blumberg2014-um,nishiura2015MERSeuro,kucharski2015MERSsuperspread, endo2020COVID, trankiem2024PNAS, hodcroft2025COVIDseqclust}; this was facilitated by the fact that the explicit form of the probability distribution for a negative binomial offspring is known and can be directly used for likelihood-based statistical inference. 

In real-world applications in epidemiology and biology, the type transition of hosts, cells, or even events may occur; transmission of infectious diseases across different age groups, or, cell differentiation into several types of cell with distinct mutational profile are among such examples. In the context of GWP, these processes are better represented by multi-type GWPs, possibly with varying environments that govern the offspring of each generation.  However, regarding the extension of GWPs from single-type to multi-type specifications, application of the total progeny distribution exists but is limited for multi-type branching processes \citep{Brummitt2012CascadeFSD, Qi2017Cascade_symbolic, kucharski2015PCB}, despite that the eventual extinction or containment probabilities for multi-type GWPs has been a key focus in an epidemiological context \citep{Spencer2011HHcontain, Evans2026-pa}. There are a few approaches that can be found in the literature, such as tree-enumeration-based methods \citep{kucharski2015PCB, chaumont2015MultitypeForests} or algebraic evaluation of the total progeny PGF, whereas the most generalizable approach is numerical coefficient extraction using Cauchy integrals, which is the focus of the present study 
(in an epidemiological context, see Miller \citeyearpar{Miller2018PGFprimer} and Roberts et al. \citeyearpar{ Roberts2023EthicalControl_sciadv}). 
Upon applying the Cauchy integral, the unique identification of the solution to the implicit equation that gives the finite total progeny PGF on the contour is required. To the best of our knowledge, this has not been explicitly discussed. The key motivation for the present study is to clarify this issue. 

Motivated by this background, we present our results on the total progeny distribution of multi-type GWPs that are relevant to computation upon real-world applications. The rest of the paper is organized as follows. Section~\ref{sec:gwp_fs_intro} introduces the total progeny PGF and its Cauchy coefficient representation. In Section~\ref{sec:fse_unique_sol}, we provide a probabilistic argument based on Good \citeyearpar{Good1955jointdistrgen} regarding the construction of the total progeny PGF, and then clarify the sufficient condition for unique identification of the solution to the total progeny fixed-point equation that corresponds to the finite total progeny PGF. Section~\ref{sec:consequences} describes consequences of the results in Section~\ref{sec:fse_unique_sol} in relation to the numerical coefficient extraction procedure. Section~\ref{sec:demonstration} compares the calculation with the Lagrange--Good formula and illustrates the finite-size likelihood surface. Section~\ref{subsec:mers} gives a likelihood-based analysis with incomplete type information.

\section{Multi-type Galton--Watson processes and total progeny PGFs}
\label{sec:gwp_fs_intro}

\subsection{Standard facts on the offspring distribution and extinction probability}
The PGF of the offspring originating from a type-$i$  individual, $i=1,2,\ldots,n$, of a multi-type GWP is generally written as 
\begin{equation}
G^{(i)}(\bm{z})=G^{(i)}(\bm{z};\bm \theta)
=\mathbb E\left[\prod_{j=1}^n z_j^{X_{ij}}\right], 
\label{eq:Gi_def}
\end{equation}
where $\bm z=(z_1,\ldots,z_n)^\top\in\mathbb C^n$, $|\bm z|\leq \bm 1$, $\bm X^{(i)}=(X_{i1},\ldots,X_{in})\in\mathbb N_0^n$ and $\bm\theta\in\mathbb R^p$ denotes a parameter vector defining the offspring distribution. Let
\[
\bm G(\bm z)=\left(G^{(1)}(\bm z),\ldots,G^{(n)}(\bm z)\right)^\top,
\qquad
\bm K=\bm J_{\bm G}(\bm 1)
=\left[\frac{\partial G^{(i)}}{\partial z_j}(\bm 1)\right]_{i,j}.
\]
We assume that the mean offspring matrix $\bm K=\bm J_{\bm G}(\bm 1)$ has finite entries. By writing the spectral radius as $\rho(.)$, the process can be classified either as subcritical, critical, or supercritical when $\rho(\bm K)<1$, $\rho(\bm K)=1$, or $\rho(\bm K)>1$, respectively.  The extinction probability is written as a vector $\bm q=(q^{(1)},\ldots,q^{(n)})^\top \in [0, 1]^n$. Among types $i = \{1, 2, \ldots, n\}$, we define $\mathcal S=\{i:q^{(i)}=0\}$ and its complement $\mathcal S^c$. 

\begin{lemma}
\label{lem:types_no_extinction}
An individual of a type in $\mathcal S$ has at least one child of a type in $\mathcal S$ with probability one. Therefore, the total number of individuals of types in $\mathcal S$ is nondecreasing over generations.  
\end{lemma}
The proof is given in Harris 1963 without any assumption on communication between types. We use this distinction between $\mathcal S$ and $\mathcal S^c$ in Lemma~\ref{lem:jacobian_reducible} when evaluating $\bm J_{\bm G}(\bm r\odot\bm q)$. Harris \citeyearpar[Chapter~II, Lemma~6.1 and Theorem~6.1]{harris1963theory} used the same observation to show transience for a positive regular, non-singular process; a brief proof sketch for the following Lemma \ref{lem:transience} for an irreducible process is proved in Appendix~\ref{app:transience}.
\begin{lemma}
\label{lem:transience}
Every finite nonzero state is transient when the process is irreducible and non-singular ($\Pr(\sum_j X_{ij}=1)<1$ for at least one type $i$). 
\end{lemma}

Regarding extinction, for the general extinction criterion, a \emph{final class} is a communicating class of types in which every parent has exactly one child in the class with probability one. Other children may have types outside the class. A \emph{final group} is a final class whose parents have no children outside the class \citep[Chapter~II, Definitions~10.1--10.2]{harris1963theory}.
\begin{lemma}
\label{lem:q_by_criticality}
$\bm q =\lim_{m \to \infty} \bm G^{\circ m}(\bm 0) \in [0, 1]^n$ is the smallest nonnegative solution of $\bm q=\bm G(\bm q)$. In general, $\bm q=\bm1$ if and only if $\rho(\bm K)\leq1$ and the process has no final class \citep[Chapter~II, Theorem~10.1]{harris1963theory}. When the process is irreducible and nonsingular,  $\bm q=\bm1$ and $\rho(\bm J_{\bm G}(\bm q))=\rho(\bm K)$ when subcritical or critical. When supercritical, $q^{(i)}<1$ for every $i$ and $\rho(\bm J_{\bm G}(\bm q))<1$.  
\end{lemma}
As in Harris 1963, at criticality a final class can persist even if it produces offspring outside the class. More generally, for an initial type $i$, the same extinction criterion applies to the types reachable from $i$, including $i$ itself, and the corresponding submatrix of $\bm K$. Thus a reducible process may have both $q^{(i)}=1$ and $q^{(j)}<1$. Appendix~\ref{app:criticality} gives the proof for the irreducible, nonsingular case.

\begin{lemma}
\label{lem:jacobian_at_extinction}
$\rho(\bm J_{\bm G}(\bm q))\leq1$ holds regardless of criticality or communication between types.
\end{lemma}
\begin{proof}
Let $\bm A=\bm J_{\bm G}(\bm q)$. As the order of types is arbitrary, here let the types in $\mathcal U=\{i:q^{(i)}<1\}$ come before those in $\mathcal E=\{i:q^{(i)}=1\}$.  We write $\bm A$ by blocks as follows:
\[
  \bm A=
  \begin{pmatrix}
    \bm A_{\mathcal U\mathcal U} & \bm A_{\mathcal U\mathcal E}\\
     \bm A_{\mathcal E\mathcal U}                         & \bm A_{\mathcal E\mathcal E}
  \end{pmatrix}.
\]
First, a parent of a type in $\mathcal E$ can only have children of types in $\mathcal E$; otherwise, $q^{(i)}<1$ follows, which is a contradiction. Thus $\bm A_{\mathcal E\mathcal U}=\bm0$ and by considering the sub-process restricted to $\mathcal E$ , $\bm q_{\mathcal E}=\bm1$ and $\bm A_{\mathcal E\mathcal E}= {\bm J_{\bm G}(\bm 1)}_{\mathcal E\mathcal E} =\bm K_{\mathcal E\mathcal E}$. Also, $\rho(\bm K_{\mathcal E\mathcal E})\leq1$ by Lemma~\ref{lem:q_by_criticality}.

Next, $\bm q$ and $\bm1$ are fixed points of $\bm G$, and elements of $\bm J_{\bm G}$ are nondecreasing functions of $\bm z$. Therefore, regarding $\bm d=\bm1-\bm q$ the following holds:
\[
  \bm d=\bm G(\bm1)-\bm G(\bm q)
       =\left(\int_0^1 \bm J_{\bm G}(\bm q+t\bm d)\,\mathrm{d}t\right)\bm d
       \geq \bm A\bm d .
\]
Since $\bm d_{\mathcal E}=\bm0$, this gives $\bm A_{\mathcal U\mathcal U}\bm d_{\mathcal U}\leq\bm d_{\mathcal U}$. As $\bm d_{\mathcal U}>\bm0$, we have $\rho(\bm A_{\mathcal U\mathcal U})\leq1$. It follows from the block triangular form
\[
  \bm A=
  \begin{pmatrix}
    \bm A_{\mathcal U\mathcal U} & \bm A_{\mathcal U\mathcal E}\\
    \bm0                         & \bm K_{\mathcal E\mathcal E}
  \end{pmatrix}
\]
that $\rho(\bm A)=\max\{\rho(\bm A_{\mathcal U\mathcal U}),
\rho(\bm K_{\mathcal E\mathcal E})\}\leq1$.
\end{proof}
For an irreducible, nonsingular process, Lemma~\ref{lem:q_by_criticality} gives $\rho(\bm J_{\bm G}(\bm q))<1$ in the subcritical and supercritical cases. At criticality, $\bm q=\bm1$ and $\rho(\bm J_{\bm G}(\bm q))=\rho(\bm K)=1$.

\subsection{PGF of the finite total progeny}
\label{sec:PGF_final_size}
For $\bm d\in\mathbb N_0^n$, let $|\bm d|_1=\sum_{j=1}^n d_j$ and $\bm z^{\bm d}=\prod_{j=1}^n z_j^{d_j}$. We follow the convention $0^0=1$ hereafter. Inequalities and absolute values of vectors are interpreted componentwise. 

Let $\bm Z_m^{(i)}$ be the offspring counts on the $m$-th generation for a process originating from $\bm Z_0^{(i)}=\bm e_i$. We define 
\[
\bm S_m^{(i)}:=\sum_{\ell=0}^m\bm Z_\ell^{(i)},
\qquad
\tau^{(i)}:=\inf\{m\geq0:\bm Z_m^{(i)}=\bm0\}, 
\]
where $\bm S_m^{(i)}$ is the cumulative offspring counts up to $m$-th generation, and $\tau^{(i)}$ is the first generation at which the process reaches extinction. For $\tau^{(i)}<\infty$ we define the finite final-size vector
\[
\bm Y^{(i)}:=\bm S_{\tau^{(i)}-1}^{(i)}
=\sum_{m=0}^{\tau^{(i)}-1}\bm Z_m^{(i)}.
\]
Here, let $\bm H (\bm z) = \left(  H ^{(1)}(\bm z),  H ^{(2)}(\bm z), \ldots, H ^{(n)}(\bm z) \right)^{\top}$ be the finite total progeny PGF of the process: 
\[
H^{(i)}(\bm z)
=\mathbb E_i[\bm z^{\bm Y};\tau<\infty]
=\sum_{\bm d\in\mathbb N_0^n}p^{(i)}(\bm d)\bm z^{\bm d},
\qquad
p^{(i)}(\bm d)=\mathbb P_i(\bm Y=\bm d,\tau<\infty),
\]
where $\mathbb P_i$ and $\mathbb E_i$ refer to a process started from $\bm Z_0=\bm e_i$. As in Harris \citeyearpar{harris1963theory}, first, $\bm H(\bm1)=\bm q$ by definition. In particular, if $i\in\mathcal S$, then $H^{(i)}(\bm z)=0$ inside the closed unit polydisc. Second, the total progeny PGF $\bm H (\bm z)$ satisfies the functional fixed-point equation, or, the total progeny equation \citep{harris1963theory}:
\begin{equation}
\bm{H}(\bm{z}) = \bm{z} \odot\bm{G}(\bm{H}(\bm{z})). 
\label{eq:fse_multi}
\end{equation}
The total progeny distribution, or, the coefficients of $\bm H (\bm z)$ for  finite indices $\bm d=(d_1,\ldots,d_n)\in\mathbb N_0^n$ can be computed by the generally applicable Cauchy integral method as 
\begin{equation}
\label{eq:LGformula}
\begin{split}
p^{(i)}(\bm d)
&= \frac{1}{(2\pi \mathrm{i})^{n}} \oint_{\Gamma_{\bm{z}}} \frac{H^{(i)}(z_{1}, \ldots, z_{n})}{z_{1}^{d_{1}+1} \cdots z_{n}^{d_{n}+1}} \,\mathrm{d}z_{1}\cdots\mathrm{d}z_{n} \\
 &= \frac{1}{(2\pi \mathrm{i})^{n}} \oint_{\Gamma_{\bm{h}}} h_{i}\left[ \prod_{m} \left( \frac{G^{(m)}(h_{1}, \ldots, h_{n})}{h_{m}} \right)^{d_{m}+1} \right] \det(\bm{J}_{\bm z \to \bm h}(\bm{h})) \,\mathrm{d}h_{1}\cdots\mathrm{d}h_{n},
\end{split}
\end{equation}
where $\Gamma_{\bm{z}}=C_{z_{1}} \times C_{z_{2}} \times \cdots \times C_{z_{n}}$ and $\Gamma_{\bm{h}}= C_{h_{1}} \times C_{h_{2}} \times \cdots \times C_{h_{n}}$ are contours with respect to complex variables $\bm{z}=(z_{1},z_{2},\ldots,z_{n})$ and $\bm{h}=(h_{1},h_{2},\ldots,h_{n})$, respectively.  The latter integral is from the Lagrange--Good formula \citep{good1960LGgeneral} that follows from the analytic change of variables $z_i=h_i/G^{(i)}(\bm h)$ that is valid near the origin when $G^{(i)}(\bm0)>0$ for every $i$, where 
\begin{equation}
\bm{J}_{\bm z \to \bm h}(\bm{h})= \left[ \frac{\partial z_{i}}{\partial h_{j}} \right] = \operatorname{diag}\left(\frac{1}{\bm G(\bm h)}\right)\left(\bm I-  \operatorname{diag}\left(\frac{\bm h}{\bm G(\bm h)} \right)\bm J_{\bm G}(\bm{h})\right).
\end{equation}

The conceptual and practical representation of $p^{(i)}$ is mathematically very straightforward for both the first and second Cauchy integral representations in Eq. \eqref{eq:LGformula}. However, the prerequisite is that, on the contour for the Cauchy integral,  we know the solution to Eq. \eqref{eq:fse_multi} that correctly represents the values of the finite total progeny distribution PGF.  In the present study, we focus on the direct computation of first integral in Eq. \eqref{eq:LGformula}; discussion on this identification issue will follow. 
\section{Unique identification of the total progeny PGF  }
\label{sec:fse_unique_sol}

Here, we construct the total progeny PGF based on Good \citeyearpar{Good1955jointdistrgen}, which essentially corresponds to tree enumeration over generations. Assuming that the offspring distribution is the same over generations, we introduce the PGF
\[
\Phi_m^{(i)}(\bm z_0,\ldots,\bm z_m)
=\mathbb E_{i}\left[\prod_{\ell=0}^m
\bm z_\ell^{\bm Z_\ell^{(i)}}  \right ],  
\]
which represents a process originating from a type-$i$ individual, and for each generation $m$ a variable $\bm z_m$ is assigned. $\Phi_m^{(i)}(\bm z_0,\ldots,\bm z_m)$ is constructed as follows by a tree extension procedure. Briefly, starting from $\bm\Phi_0(\bm z_0)=\bm z_0$, tree extension from $m$-th to $m+1$-th generation corresponds to the substitution $\bm z_{m} \to \bm z_{m} \odot \bm G(\bm z_{m+1})$. By defining 
\[\bm T_{\bm a}(\bm b):=\bm a\odot\bm G(\bm b) \] for arbitrary vectors $\bm a$ and $\bm b$, this tree extension process is written as
\begin{align*}
\bm\Phi_0(\bm z_0)
&=\bm z_0, \\
\bm\Phi_1(\bm z_0,\bm z_1)
&=\bm T_{\bm z_0}(\bm z_1)
=\bm z_0\odot\bm G(\bm z_1),\\
\bm\Phi_2(\bm z_0,\bm z_1,\bm z_2)
&=\left( \bm T_{\bm z_0}\circ\bm T_{\bm z_1}
\right)(\bm z_2)
=\bm z_0\odot\bm G\!\left(
  \bm z_1\odot\bm G(\bm z_2)
\right), 
\end{align*}
and in general,    
\[
\bm\Phi_m(\bm z_0,\ldots,\bm z_m) =  \left( \bm T_{\bm z_0}\circ\bm T_{\bm z_1}\circ\cdots
  \circ\bm T_{\bm z_{m-1}}
\right)(\bm z_m) , 
\]
where 
\[
{\bm\Phi_m(\bm z_0,\ldots,\bm z_m)}^{(i)} = \mathbb E_i[\bm z_0^{\bm Z_0^{(i)}} \bm z_1^{\bm Z_1^{(i)}}\bm z_2^{\bm Z_2^{(i)}} \cdots \bm z_{m}^{\bm Z_{m}^{(i)}}] .
\]
To distinguish the last generation from the preceding ones, for $m\geq1$ we define
\[
\bm{\mathcal H}_m(\bm z,\bm u)
:=\bm\Phi_m(\underbrace{\bm z,\ldots,\bm z}_{m\ \mathrm{arguments}},\bm u)
=\bm T_{\bm z}^{\circ m}(\bm u),
\qquad
\bm{\mathcal H}_0(\bm z,\bm u):=\bm u.
\]
Thus
\[
\mathcal H_m^{(i)}(\bm z,\bm u)
=\mathbb E_i\!\left[\bm z^{\bm S_{m-1}}\bm u^{\bm Z_m}\right].
\]
Note that variable $\bm u$ represents the terminal generation, whereas $\bm z$ represents the total offspring that happened before the $m$-th generation.  

Regarding $\bm {\mathcal H}_m(\bm z,\bm u)$, by letting $\bm u = \bm 0$ we have the following. 
\begin{rem}
\label{rem:pgf_end_before_m}
\[
{\mathcal H}^{(i)}_{m}(\bm z, \bm 0)= \left( \bm T_{\bm z} ^{\circ m} (\bm 0) \right)^{(i)} = \mathbb E_i[\bm z^{\bm S_{m-1}^{(i)}}\bm 0^{\bm Z_{m}^{(i)}}] 
=\mathbb E_i \left[\bm z^{\bm S_{m-1}^{(i)}}; \tau^{(i)}\leq m\right]
\]
is the PGF for the total offspring in trees that become extinct by generation $m$. For $|\bm z| \leq \bm 1$, at the asymptotics $m \to \infty$ we have 
\[
\lim_{m\to \infty}  {\mathcal H}^{(i)}_{m}(\bm z, \bm 0) 
= \lim_{m\to \infty} \left( \bm T_{\bm z} ^{\circ m} (\bm{0})\right)^{(i)}
=  H^{(i)}(\bm z) \]
with $\bm H(\bm 1) = \bm q$, regardless of criticality. 
\end{rem}
The next lemma is a generalization of Remark \ref{rem:pgf_end_before_m} , where the terminal marker $\bm u$ lies inside the closed unit polydisc when $|\bm z|<\bm1$.
\begin{lemma}
\label{lem:converge_to_H}
For $| \bm z | < \bm 1$ and $| \bm u | \leq \bm 1$,  
\[
\bm{\mathcal H}_m(\bm z,\bm u)
=\bm T_{\bm z}^{\circ m}(\bm u)
\to\bm H(\bm z)
\]
as $m\to\infty$ .
\end{lemma}
\begin{proof}
Suppose we have two vectors $\bm x,  \bm y \in \mathbb C^n$ such that $| \bm x| , | \bm y | \leq  \bm 1$. Then the generation-wise representation gives 
\begin{align}
(\bm T_{\bm z}^{\circ m}(\bm x) )^{(i)} - (\bm T_{\bm z}^{\circ m}(\bm y ) )^{(i)} 
&= 
{\mathcal H}_m^{(i)}(\bm z,\bm x)
 -{\mathcal H}_m^{(i)}(\bm z,\bm y)  \\
 &= 
 \mathbb E_i\!\left[
  \bm z^{\bm S_{m-1}^{(i)}}
  \left( {\bm x}^{\bm Z_m^{(i)}}-{\bm y}^{\bm Z_m^{(i)}}\right)
\right] \\
&=\sum_{\substack{\bm N,\bm M\in\mathbb N_0^n\\\bm N\neq\bm0}}\Pr(\bm S_{m-1}^{(i)}=\bm M,\bm Z_m^{(i)}=\bm N)\bm z^{\bm M}(\bm x^{\bm N}-\bm y^{\bm N}),
\label{eq:nonconstant_terminal_part_prob}
\end{align}
for $i = 1, 2, \ldots, n$, where terms with $\bm Z_{m}^{(i)} = \bm 0$ cancel out. Because $ |{\bm S_{m-1}^{(i)}}|_1 \geq m$ holds for terms with $\bm Z_m ^{(i)} \neq\bm0$, for $|\bm z | < \bm 1$ it follows that
\begin{equation}
\left| 
(\bm T_{\bm z}^{\circ m}(\bm x) )^{(i)} - (\bm T_{\bm z}^{\circ m}(\bm y ) )^{(i)}
\right|
=
\left|
{\mathcal H}_m^{(i)}(\bm z,\bm x)
-{\mathcal H}_m^{(i)}(\bm z,\bm y)
\right|
\leq 2{r_{\max}}^m, 
\label{eq:terminal_generation_decay_prob}
\end{equation}
where $r_{\max} = \max_j |z_{j}|$. Taking $\bm x=\bm0$ and $\bm y=\bm u$, the right-hand side of Eq.~\eqref{eq:terminal_generation_decay_prob} tends to zero. The result follows from Remark~\ref{rem:pgf_end_before_m}.
\end{proof}
Based on the argument so far, we introduce the following theorem that relates the numerical solutions of $\bm h$ for $\bm h=\bm T_{\bm z}(\bm h)=\bm z\odot\bm G(\bm h)$ as a function of $\bm z$ to the values of $\bm H(\bm z)$. Note that the irreducibility assumption is not required, because $|\bm z|<\bm 1$.
\begin{thm}
\label{thm:sol_region}
For every $|\bm z| < \bm 1$, there is only a single solution $\bm h^* = \bm h^*(\bm z)$ to the fixed-point equation
\[
\bm h=\bm T_{\bm z}(\bm h)=\bm z\odot\bm G(\bm h)
\]
in the closed unit polydisc $\{\bm h\in\mathbb C^n:|\bm h| \leq \bm 1\}$, which is identical to the values of the total progeny PGF $\bm H(\bm z)$ at every $\bm z$. Also, we have
\[
|\bm h^*|\leq |\bm z|\odot\bm q.
\]
\end{thm}
\begin{proof}
From Lemma \ref{lem:converge_to_H}, $\lim_{m \to \infty} \bm T_{\bm z}^{(\circ m)}(\bm y)=\bm h^*(\bm z)$ is equal to the value of $\bm H(\bm z)$ whenever $|\bm y|\leq\bm1$. In particular, if $\bm y$ is a fixed point, then $\bm T_{\bm z}^{\circ m}(\bm y)=\bm y$ for every $m$, and the preceding limit gives $\bm y=\bm H(\bm z)$. This proves uniqueness. Also, because $\bm h^*(\bm z)= \bm H (\bm z)$ is a fixed point satisfying $\bm h^* = \bm T_{\bm z}(\bm h^*)$, it follows that
\[
|\bm h^*|
\leq |\bm z|\odot\bm G(|\bm h^*|)
\leq |\bm z|\odot\bm G(\bm q)
=|\bm z|\odot\bm q. 
\]
\end{proof}
\begin{rem}
Remark~\ref{rem:pgf_end_before_m} gives convergence for $|\bm z|\leq\bm1$ when $\bm u=\bm0$. For $|\bm z|<\bm1$, Lemma~\ref{lem:converge_to_H} allows every $|\bm u|\leq\bm1$. At $\bm z=\bm1$ a supercritical process has at least two distinct fixed points, $\bm q$ and $\bm1$.
\end{rem}

\section{Consequences of Theorem \ref{thm:sol_region}}
\label{sec:consequences}

We discuss a few issues that relate to numerical analysis and statistical inference. 

\subsection{An error bound for Picard iteration}
\begin{cor}
\label{cor:error_bound_picard}
For $|\bm z|<\bm1$, let $r_{\max}=\max_j|z_j|$, and $|\bm u|\leq\bm1$. Then, for every initial type $i$, 
\[
\left|
\left(\bm T_{\bm z}^{\circ m}(\bm u)\right)^{(i)}
-H^{(i)}(\bm z)
\right|
\leq2r_{\max}^m, 
\]
and for $\bm u=\bm 0$, 
\[
\left|
\left(\bm T_{\bm z}^{\circ m}(\bm0)\right)^{(i)}
-H^{(i)}(\bm z)
\right|
\leq r_{\max}^m.
\]
\end{cor}
The first inequality in Corollary \ref{cor:error_bound_picard} follows by taking $\bm x=\bm H(\bm z)$ and $\bm y=\bm u$ in Eq.~\eqref{eq:terminal_generation_decay_prob}. The second inequality in Corollary\ref{cor:error_bound_picard} corresponds to a special case $\bm y=\bm0$ in Eq.~\eqref{eq:nonconstant_terminal_part_prob}.  

\subsection{Varying environments}
\label{subsec:varying_environment}
Offspring laws that depend on generation arise in some biological models. \cite{Haccou2003-dp}. A process in which the offspring law changes by generation is a special case of the cascade process as described in Good \cite{Good1955jointdistrgen} ; 
\begin{cor}[Varying environment]
\label{cor:varying_environment}
Let $\bm \Omega=(\bm \omega_0,\bm \omega_1,\ldots)$ be a fixed sequence of parameters, and let $\bm G_{\ell}$ be the offspring PGF with parameter $\bm \omega_{\ell}$ for a finite-type branching process, where individuals in each generation are independent between each other. Let
\[
\bm T_{\ell,\bm z}(\bm h)=\bm z\odot\bm G_\ell(\bm h),
\qquad
\bm{\mathcal H}_{\bm \Omega,m}(\bm z,\bm u)
=\bigl(\bm T_{0,\bm z}\circ\cdots\circ
\bm T_{m-1,\bm z}\bigr)(\bm u),
\]
and let $H_{\bm \Omega}^{(i)}(\bm z)=\mathbb E_{i,\bm\Omega}[\bm z^{\bm Y};\tau<\infty]$ be the finite-total-progeny PGF for $\bm \Omega$. For $|\bm z|<\bm1$, $r_{\max}=\max_j|z_j|$, $|\bm u|\le\bm1$, and $m\ge1$,
\begin{equation}
\left|\mathcal H_{\bm \Omega,m}^{(i)}(\bm z,\bm u)
-H_{\bm \Omega}^{(i)}(\bm z)\right|
\le 2r_{\max}^{m}.
\label{eq:varying_env_terminal_bound}
\end{equation}
and if $\bm u = \bm 0$, 
\begin{equation}
\left|\mathcal H_{\bm\Omega,m}^{(i)}(\bm z,\bm0)
-H_{\bm \Omega}^{(i)}(\bm z)\right|
\le r_{\max}^{m}.
\label{eq:varying_env_zero_bound}
\end{equation}
\end{cor}
\begin{proof}
The generation-wise construction gives
\[
\mathcal H_{\bm \Omega,m}^{(i)}(\bm z,\bm u)
=\mathbb E_{i,\bm \Omega}
\left[\bm z^{\bm S_{m-1}}\bm u^{\bm Z_m}\right].
\]
By definition, terms corresponding to  $\tau \leq m$ should agree between $\mathcal H_{\bm \Omega,m}^{(i)}(\bm z,\bm u)$ and $H_{\bm \Omega}^{(i)}(\bm z)$.  Therefore, 
\begin{align*}
\mathcal H_{\bm \Omega,m}^{(i)}(\bm z,\bm u)
-H_{\bm \Omega}^{(i)}(\bm z)
={}&\mathbb E_{i,\bm \Omega}
\left[\bm z^{\bm S_{m-1}}\bm u^{\bm Z_m};\tau>m\right]\\
&-\mathbb E_{i,\bm \Omega}
\left[\bm z^{\bm Y};m<\tau<\infty\right].
\end{align*}
Note that $|\bm S_{m-1}|_1\ge m$ and $|\bm Y|_1\ge m$ as $\tau >m$, so $r_{\max}^m$ is the apparent upper bound for both terms. Especially when $\bm u=\bm0$, the first expectation vanishes. 
\end{proof}

For an exogenous random environment, or a space of environment that contains $\bm \Omega$, it is immediate that expectation over all possible paths $\bm \Omega$ gives 
\[
\left|\mathbb E_{\bm \Omega}
\left[\mathcal H_{\bm \Omega,m}^{(i)}(\bm z,\bm0)\right]
-\mathbb E_{\bm \Omega}\left[H_{\bm \Omega}^{(i)}(\bm z)\right]\right|
\le r_{\max}^m, 
\]
though $H_{\bm \Omega}(\bm z)$ does not necessarily match between each other for different $\bm \Omega$. 

\subsection{Sharper contraction inside an invariant polydisc }
Corollary~\ref{cor:error_bound_picard} gives an error bound without irreducibility. For a fixed $\bm z\in\mathbb C^n$ such that $|\bm z|=\bm r<\bm1$, let $\mathcal D_{\bm z}:=
\{\bm h\in\mathbb C^n:|\bm h|\leq|\bm z|\odot\bm q\}$. As we see that $\bm H(\bm z) \in \mathcal D_{\bm z}$ from Theorem \ref{thm:sol_region} , we obtain a sharper contraction regarding $\bm T_{\bm z}(\bm h)$ inside $\bm h \in \mathcal D_{\bm z}$. 

We first evaluate $\bm J_{\bm G}(\bm r\odot\bm q)$ using the offspring of parents in $\mathcal S$ and $\mathcal S^c$. We use the resulting block bounds in the proof of Lemma~\ref{lem:contraction}. Note that $\mathcal D_{\bm z} = \{\bm 0\}$ when $\mathcal S^c$ is empty.
\begin{lemma}
\label{lem:jacobian_reducible}
Let $\bm A=\bm J_{\bm G}(\bm r\odot\bm q)$, where $\bm0\leq\bm r<\bm1$, and let $\mathcal S=\{i:q^{(i)}=0\}$. If $\mathcal S^c$ is nonempty, then $\rho(\bm A_{\mathcal S^c\mathcal S^c})<1$. For the full matrix, $\rho(\bm A)\leq1$, with equality if and only if the process has a final group.
\end{lemma}
\begin{proof}
For a subcritical process, $\rho(\bm A)\leq\rho(\bm K)<1$ because $\bm A\leq\bm K$ componentwise. Hereafter, the argument will be focused on critical and supercritical processes. 

First consider nonempty $\mathcal S^c$. Let $\bm q_{\mathcal S^c}$ be a vector of extinction probabilities for types $i \in \mathcal S^c$. By considering the functional iteration giving $\bm G^{\circ m}(\bm0) \to \bm q$, we have $\left( G^{\circ m}(\bm0) \right)^{(i)} = 0$ for $i \in \mathcal S$ and all $m$. For $i \in \mathcal S^c$, let $F^{(i)}(\bm y)=G^{(i)}(\bm0_{\mathcal S},\bm y)$ be the generating function, which represents the subprocess of $\bm G$ where only offspring such that all children are in types $i \in \mathcal S^c$. By definition, $F^{(i)}(\bm1)=\Pr(\sum_{j\in\mathcal S}\bm Z^{(i)}_{j}=0)\leq1$, and $\bm F^{\circ m}(\bm0) \to \bm q_{\mathcal S^c}$ as $m \to \infty$. 

Along the linear path from $\bm z =\bm r \odot \bm q$ to $\bm z =\bm q$, all terms in $\bm G(\bm z)$ that correspond to offspring with at least one child in $\mathcal S$ vanishes, even after differentiation by variables $z_j, j \in \mathcal S^c$. Therefore, we focus on the Jacobian $\bm B=\bm A_{\mathcal S^c\mathcal S^c}
=\bm J_{\bm F}(\bm r_{\mathcal S^c}\odot\bm q_{\mathcal S^c})$.  Given that $\bm B \leq  \left(\bm J_{\bm G}(\bm q)\right)_{\mathcal S^c \mathcal S^c}$ componentwise,  Lemma~\ref{lem:jacobian_at_extinction} gives $\rho(\bm B)\leq1$. Here, suppose that $\rho(\bm B)=1$ and let $\bm v$ be a nonzero, nonnegative right Perron vector of $\bm B$. Then, there exists $t>0$ sufficiently small such that $\bm y=\bm q_{\mathcal S^c}-t\bm v
\geq\bm r_{\mathcal S^c}\odot\bm q_{\mathcal S^c}$. Path integration along the line from $\bm z = \bm y$ to $\bm z =  \bm q_{\mathcal S^c}$ gives
\[
  \bm F(\bm q_{\mathcal S^c})-\bm F(\bm y)
  \geq t\bm J_{\bm F}(\bm y)\bm v
  \geq t\bm B\bm v=t\rho(\bm B)\bm v = t\bm v, 
\]
or by rearrangement, 
\[
\bm q_{\mathcal S^c}- t\bm v \geq \bm F(\bm y). 
\]
This contradicts $\bm F^{\circ m}(\bm0)\to \bm q_{\mathcal S^c}$; therefore, $\rho(\bm A_{\mathcal S^c\mathcal S^c})<1$.

Next, we consider $\mathcal S$. If $\mathcal S$ is empty, $\rho(\bm A)<1$ from the preceding argument. Otherwise, for a parent of type $i\in\mathcal S$, let $N_{\mathcal S}=\sum_{j\in\mathcal S}X_{ij}$ be the number of its children in $\mathcal S$. By Lemma~\ref{lem:types_no_extinction}, $N_{\mathcal S}\geq1$; therefore, differentiating $G^{(i)}$ by variables $z_j, j \in \mathcal S^c$ leaves a factor from $\mathcal S$ in every term. This factor vanishes at $\bm r\odot\bm q$, so for the offdiagonal block of $A$ we have $\bm A_{\mathcal S\mathcal S^c}=\bm0$. 

Regarding the diagonal block $\bm A_{\mathcal S \mathcal S}$,  for an arbitrary type $i \in \mathcal S$, differentiating $\bm G^{(i)}$ by variables $z_j, j \in \mathcal S$ results only leaves nonzero terms for $N_{\mathcal S}=1$, or, the linear terms in the original offspring PGF; for terms with $N_{\mathcal S}\geq 2$, elements with at least order one remains in $\bm A_{\mathcal S \mathcal S}$, which vanish along the linear path from $\bm z =\bm r \odot \bm q$ to $\bm z =\bm q$. In addition, regardless of $N_{\mathcal S}=1$, the factor $\prod_{k\in\mathcal S^c}(r_kq^{(k)})^{X_{ik}}<1$ remains for terms representing type $i \in \mathcal S$ having a child in $\mathcal S^c$, so in such case $\sum_{j\in\mathcal S} A_{ij}<1$. In summary,  for $i\in\mathcal S$,
\[
 \sum_{j\in\mathcal S}A_{ij}
 =\mathbb E\!\left[\mathbf1_{\{N_{\mathcal S}=1\}}
       \prod_{k\in\mathcal S^c}(r_k q^{(k)})^{X_{ik}}\right]\leq1.
\]
Equality holds exactly when the parent has a single child belonging to $\mathcal S$.

Let $\bm D=\bm A_{\mathcal S\mathcal S}$. The preceding argument shows that $\bm D$ is substochastic. Thus, $\lambda = \rho(\bm D)\leq1$. Specifically, if the process has a final group $\mathcal F$, then $\mathcal F\subseteq\mathcal S$ and $\bm D_{\mathcal F\mathcal F}$ is stochastic. Therefore, $1=\rho(\bm D_{\mathcal F\mathcal F})
  \leq \lambda \leq1$. 

In case that there is not final group, then $\rho(\bm D)<1$ can be shown as follows. Let $\mathcal S^{\text{sub}} = \{i\in S :\sum_{j\in \mathcal S} D_{ij} <1\}$, and $\bm v$ be a nonzero, nonnegative right Perron vector of $\bm D$, and 
\[
  M=\max_{j\in\mathcal S}v_j>0,
  \qquad
  \mathcal V=\{i\in\mathcal S:v_i=M\}.
\]
From $\bm D \bm v = \lambda \bm v$, for every $i\in\mathcal V$,
\[
  \lambda M=\sum_{j\in\mathcal S}D_{ij}v_j
  \leq M\sum_{j\in\mathcal S}D_{ij}
  \leq M.
\]
If $\mathcal S^{\text{sub}} \cap \mathcal V \neq \emptyset$, then strict inequality $\lambda <1$ follows from above. Else, $\sum_{j\in \mathcal S} D_{ij}=1$ holds for all $i \in \mathcal V$. In this case, from

\[
  \lambda M = \sum_{j\in\mathcal V}D_{ij}M + \sum_{j\in\mathcal V^c}D_{ij}v_j
\]
we see that $\lambda <1$ when $\sum_{j\in\mathcal V^c}D_{ij}>0$. Otherwise, $\sum_{j\in\mathcal V^c}D_{ij}=0$ indicates that $\bm D_{\mathcal V\mathcal V}$ is stochastic, and also the existence of a final group in $\mathcal V$. 

In summary, for
\[
  \bm A=
  \begin{pmatrix}
    \bm A_{\mathcal S\mathcal S} & \bm A_{\mathcal S\mathcal S^c}\\
     \bm A_{\mathcal S^c\mathcal S}                         & \bm A_{\mathcal S^c \mathcal S^c}
  \end{pmatrix}
  =
  \begin{pmatrix}
    \bm A_{\mathcal S\mathcal S} & \mathbf 0\\
     \bm A_{\mathcal S^c\mathcal S}                         & \bm A_{\mathcal S^c \mathcal S^c}
  \end{pmatrix}, 
\]
$\rho(\bm A) = \max \{\rho(\bm A_{\mathcal S\mathcal S}), \rho(\bm A_{\mathcal S^c \mathcal S^c})\}\leq 1$, with equality only when $\mathcal S$ has a final group.
\end{proof}

\begin{lemma}
\label{lem:contraction}
For every $|\bm z|<\bm1$, $\bm T_{\bm z}(\bm h) \in \mathcal D_{\bm z}$ for all $\bm h \in \mathcal D_{\bm z}$. Also, $\bm T_{\bm z}$ is a contraction with factor at most $r_{\max}=\max_j r_j$ with respect to a suitable weighted infinity norm. If $\mathcal D_{\bm z} \neq \{ \bm 0 \}$, the factor can be chosen strictly smaller than $r_{\max}$.   
\end{lemma}
\begin{proof}
Here we write $\bm r=|\bm z|, r_{\max} = \max_{j} r_j$.  First, if $\bm h\in\mathcal D_{\bm z}$, then
\[
|\bm T_{\bm z}(\bm h)|
\leq\bm r\odot\bm G(|\bm h|)
\leq\bm r\odot\bm G(\bm r\odot\bm q)
\leq\bm r\odot\bm G(\bm q)
=\bm r\odot\bm q,
\]
so $\bm T_{\bm z}(\bm h) \in \mathcal D_{\bm z}$. 

Next, we define the weighted infinity norm of vectors and matrices for an arbitrary positive vector $\bm x$ by $\|\bm a\|_{\infty,\bm x}=\max_i\frac{|a_i|}{x_i}, \|\bm A\|_{\infty,\bm x}=\max_i\frac{(|\bm A|\bm x)_i}{x_i}$. Suppose we have $\bm{h}, \bm{h}^{\prime}\in \mathcal D_{\bm z}$ such that $\bm{h}\ne \bm{h}^{\prime}$, and consider a linear path $\bm{\omega}(t)=\bm{h}+t(\bm{h}^{\prime}-\bm{h}), 0\le t\le1$. Along this path, 
\[
\bm T_{\bm z}(\bm h')-\bm T_{\bm z}(\bm h) = \left(\int_{t=0}^{t=1} \operatorname{diag}(\bm z)\bm J_{\bm G}(\bm\omega(t))\,\mathrm{d}t\right)(\bm h'-\bm h), 
\]
and from $|\bm\omega(t)|\leq |\bm z| \odot \bm q$ the following holds: 
\begin{align*}
\|\bm T_{\bm z}(\bm h')-\bm T_{\bm z}(\bm h)\|_{\infty,\bm x} 
& \leq
\|\operatorname{diag}(\bm r)\bm J_{\bm G}(\bm r \odot \bm q)\|_{\infty,\bm x}
\|\bm h'-\bm h\|_{\infty,\bm x} \\
&\leq
r_{\max}\|\bm J_{\bm G}(\bm r \odot \bm q)\|_{\infty,\bm x}\|\bm h'-\bm h\|_{\infty,\bm x}.
\end{align*}

If $\mathcal D_{\bm z}=\{\bm0\}$, we may take the contraction factor to be zero. Otherwise $\mathcal S^c$ is nonempty and $r_{\max}>0$.
Let $\bm A=\bm J_{\bm G}(\bm r\odot\bm q)$. For $i\in\mathcal S$, both $h'_i-h_i$ and $[\bm T_{\bm z}(\bm h')-\bm T_{\bm z}(\bm h)]_i$ vanish. Thus, for $i\in\mathcal S^c$,
\[
 \bigl|[\bm T_{\bm z}(\bm h')-\bm T_{\bm z}(\bm h)]_i\bigr|
 \leq r_i\sum_{j\in\mathcal S^c} A_{ij}|h'_j-h_j|.
\]
Let $\bm B=\bm A_{\mathcal S^c\mathcal S^c}$. By Lemma~\ref{lem:jacobian_reducible}, $\rho(\bm B)<1$, so we may take
\[
 \bm w_{\mathcal S^c}=(\bm I-\bm B)^{-1}\bm1_{\mathcal S^c}
 \geq\bm1_{\mathcal S^c}.
\]
The weights on $\mathcal S$ may be any positive numbers, for example $w_i=\epsilon$ for all $i\in\mathcal S$ and any $\epsilon>0$. Since $\bm B\bm w_{\mathcal S^c}=\bm w_{\mathcal S^c}-\bm1_{\mathcal S^c}$, it follows that
\[
 \|\bm T_{\bm z}(\bm h')-\bm T_{\bm z}(\bm h)\|_{\infty,\bm w}
 \leq r_{\max}\max_{i\in\mathcal S^c}\left(1-\frac{1}{w_i}\right)
      \|\bm h'-\bm h\|_{\infty,\bm w} < r_{\max}
      \|\bm h'-\bm h\|_{\infty,\bm w}.
\]
\end{proof}
Appendix~\ref{app:contraction_H} gives an alternative choice of weights for Lemma~\ref{lem:contraction}. Based on Lemma \ref{lem:contraction}, a numerical approach based on Newton's method becomes an option for suitable contours and will be used in later demonstrations when the process is irreducible. Details are provided in Appendix \ref{app:newton_picard_hybrid}. 

\subsection{An error bound for Cauchy coefficient extraction}
\label{subsec:coefficient_error}
Let $\Gamma_{\bm z}$ be a fixed contour with $|z_j|=r_j$ and $0<r_j<1$ for each $j$. The following relationship regarding the coefficients of $\bm H^{(i)}(\bm z)$ holds. \begin{rem}[Error in a DFT coefficient]
\label{rem:coefficient_error}
Let $\Gamma_N$ be grid on $\Gamma_{\bm z}$ that is chosen for DFT, with $z_j=r_j e^{2\pi\mathrm{i}\ell_j/N}$, $\ell_j=0,\ldots,N-1$, where $0<r_j<1$ and $r_{\max}=\max_j r_j$. Let arbitrary initial type $i$ and $\bm d\in\{0,\ldots,N-1\}^n$ be fixed. Below we write $D_{N,\bm d}$ for the DFT in Eq.~\eqref{eq:coef_dft}, including the factor $1/\bm r^{\bm d}=\prod_j r_j^{-d_j}$.  Then, 
\[
\begin{aligned}
|p-\widehat p_N|
 &\le \underbrace{|p-p_N|}_{\text{aliasing}}
      +\underbrace{|p_N-\widehat p_N|}_{\text{PGF approximation}}\\
 &\le q^{(i)}r_{\max}^{N}+\frac{\varepsilon_i}{\bm r^{\bm d}},
\end{aligned}
\]
where $p = p^{(i)}(\bm d)=[\bm z^{\bm d}]H^{(i)}(\bm z)$ is the exact coefficient, $p_N = D_{N,\bm d}[H^{(i)}]$ is the DFT of the exact PGF values, and $\widehat p_N = D_{N,\bm d}[F]$ is the DFT of the computed PGF values. 
\end{rem}

\subsection{Differentiability of the total progeny distribution}
\label{subsec:theta_differentiable}
Here we explicitly treat $\bm G(\bm z)=\bm G(\bm z;\bm\theta)$ as a function of $\bm\theta\in\mathcal P$, where $\mathcal P\subset\mathbb R^p$ is an arbitrary open region that may contain subcritical, critical, and supercritical parameter values. For every $\bm\theta\in\mathcal P$, assume that the conditions in Section \ref{sec:gwp_fs_intro} hold. Also, for an integer $s\geq1$, assume that $\partial_{\bm\theta}^{\bm\alpha}\bm G(\bm z;\bm\theta)$ exists and is continuous for $(\bm\theta,\bm z)$ for every $\bm\alpha\in\mathbb N_0^p$ with $|\bm\alpha|_1\leq s$. 

First, we clarify the coefficient argument regarding $\bm G(\bm z;\bm \theta)$, that will also be used for $\bm H(\bm z;\bm\theta)$. Briefly, for a finite index $\bm d\in\mathbb N_0^n$, derivatives ${\partial^ {\bm \alpha} p_{\bm G}^{(i)}(\bm d ; \bm \theta)}/{\partial \bm \theta ^ {\bm \alpha}}$ exist for $p_{\bm G}^{(i)}(\bm d ; \bm \theta) = [\bm z^{\bm d}] G^{(i)}(\bm z;\bm \theta)$. This follows from the differentiability assumption on $\bm G$ regarding $\bm \theta$, which leads to
\[
 \partial_{\bm\theta}^{\bm\alpha}p_{\bm G}^{(i)}(\bm d;\bm\theta)
 =\frac{1}{(2\pi\mathrm i)^n}\oint_{\Gamma_{\bm z}}
 \frac{\partial_{\bm\theta}^{\bm\alpha}
       G^{(i)}(\bm z;\bm\theta)}{\bm z^{\bm d+\bm1}}\,\mathrm{d}\bm z,
 \qquad
 p_{\bm G}^{(i)}(\bm d;\bm\theta)
 =[\bm z^{\bm d}]G^{(i)}(\bm z;\bm\theta).
\]

A similar discussion regarding $\bm H(\bm z;\bm \theta)$ follows.  
\begin{cor}[Parameter differentiability]
\label{cor:theta_differentiability}
Under the assumptions stated above, $\bm H(\bm z;\bm\theta)$ is $C^s$ with respect to $\bm\theta$ for every $|\bm z|<\bm1$. Its coefficient
\[
p_{\bm H}^{(i)}(\bm d;\bm\theta)
:=[\bm z^{\bm d}]H^{(i)}(\bm z;\bm\theta)
\]
is also $C^s$ with respect to $\bm\theta$ for every fixed finite $\bm d=(d_1,\ldots,d_n)\in\mathbb N_0^n$.
\end{cor}
\begin{proof}
For each finite $m$, the functional iterate $\bm{\mathcal H}_m(\bm z,\bm0;\bm\theta)
=\bm T_{\bm z}^{\circ m}(\bm0;\bm\theta)$ is continuous in $(\bm z,\bm\theta)$ by induction. For $|\bm z|<\bm1$, let $r_{\max}=\max_j|z_j|$. Corollary~\ref{cor:error_bound_picard} gives
\begin{equation}
 \bigl\|\bm H(\bm z;\bm\theta)
    -\bm{\mathcal H}_m(\bm z,\bm0;\bm\theta)\bigr\|_\infty
 \leq r_{\max}^m .
 \label{eq:theta_continuity_bound}
\end{equation}
The bound is uniform in $\bm\theta$ and on every closed polydisc strictly inside the unit polydisc, which proves that $\bm H$ is continuous for $\bm \theta$.

Now fix $\bm z$ and $\bm\theta$, and write $\bm h=\bm H(\bm z;\bm\theta)$ and $\bm A=\operatorname{diag}(\bm z)\bm J_{\bm G}(\bm h;\bm\theta)$. Since $|\bm h|\leq |\bm z|\odot\bm q(\bm\theta)\leq\bm q(\bm\theta)$, the nonnegative PGF coefficients and Lemma~\ref{lem:jacobian_at_extinction} give
\[
 \rho(\bm A)\leq\rho(|\bm A|)
 \leq r_{\max}\rho\!\left(\bm J_{\bm G}(\bm q(\bm\theta);\bm\theta)\right)
 \leq r_{\max}<1 .
\]
Thus $\bm I-\bm A$ is invertible. To differentiate the fixed-point equation, take $\bm\theta'=\bm\theta+\Delta\bm\theta$ near $\bm\theta$, and write $\bm h'=\bm H(\bm z;\bm\theta')$. Subtracting the two equations and integrating along the line segments in $\bm h$ and $\bm\theta$ gives
\[
  (\bm I-\bm A_\Delta)(\bm h'-\bm h)=\bm C_\Delta\,\Delta\bm\theta,
\]
where
\begin{align*}
 \bm A_\Delta&=\operatorname{diag}(\bm z)\int_0^1
    \bm J_{\bm G}(\bm h+t(\bm h'-\bm h);\bm\theta)\,\mathrm dt,\\
 \bm C_\Delta&=\operatorname{diag}(\bm z)\int_0^1
    \partial_{\bm\theta}\bm G(\bm h';\bm\theta+t\Delta\bm\theta)\,\mathrm dt .
\end{align*}
By continuity, $\bm A_\Delta\to\bm A$ and $\bm C_\Delta\to\operatorname{diag}(\bm z)\partial_{\bm\theta}\bm G(\bm h;\bm\theta)$. Hence $\bm I-\bm A_\Delta$ is invertible for all sufficiently small $\Delta\bm\theta$, and
\[
 \partial_{\bm\theta}\bm H(\bm z;\bm\theta)
   =(\bm I-\bm A)^{-1}\operatorname{diag}(\bm z)
      (\partial_{\bm\theta}\bm G)(\bm h;\bm\theta).
\]
Cauchy's formula in the interior of the unit polydisc gives continuity of all mixed derivatives of $\bm G$ in its PGF variables and up to order $s$ in $\bm\theta$. The displayed expression is continuous, and differentiating it successively shows that $\bm H$ is $C^s$ in $\bm\theta$.

Finally, a family of $|\bm d|_1$ individuals has no individuals left in generation $m$ if $m\geq |\bm d|_1$. Thus for any such $m$,
\[
 p_{\bm H}^{(i)}(\bm d;\bm\theta)
  =[\bm z^{\bm d}]\mathcal H_m^{(i)}(\bm z,\bm0;\bm\theta).
\]
This finite composition is $C^s$ in $\bm\theta$ by induction. Its parameter derivatives are jointly continuous on any fixed interior Cauchy contour, so differentiating its coefficient integral up to order $s$ proves the claim.
\end{proof}

Corollary~\ref{cor:theta_differentiability} shows that each finite-size probability $p_{\bm H}^{(i)}(\bm d;\bm \theta)$ is a $C^s$ function of $\bm \theta$, as is the underlying offspring, regardless of criticality that is manifested by $\bm \theta$ itself. This property is relevant to statistical inference of $\bm \theta$ from real-world data using an unconditional likelihood (a likelihood that is not conditioned on extinction; for supercritical cases) for the observed final-size data, because the finite-size probabilities remain smooth across the criticality threshold. 

\subsection{Likelihoods conditioned on extinction}
\label{subsec:identifiability}
Likelihoods conditioned on extinction are widely used in an epidemiological context when a supercritical process is assumed \citep{nishiura2012Supercritical}. Again we assume the conditions in Section \ref{subsec:theta_differentiable}. Let $\mathcal P^{\mathrm{sup}}\subset\mathcal P$ denote an open, supercritical parameter region. In this subsection, the additional assumptions
\begin{itemize}
    \item $q^{(i)}(\bm\theta)>0$ for every type $i$
    \item   $\bm\theta\in\mathcal P^{\mathrm{sup}}$
    \item  irreducibility
\end{itemize}
are imposed. Following the notations in preceding sections, let $\bm q = \bm q(\bm \theta) \in (0, 1)^n$ for $\bm \theta \in \mathcal P^{\mathrm{sup}}$ denote the extinction-probability vector satisfying $\bm q = \bm G(\bm q;\bm \theta)$ for supercritical processes. The conditional PGF for the offspring and the total progeny can be written as 
\[
\tilde{\bm G}(\bm z;\bm \theta) = \frac{\bm G(\bm q(\bm \theta) \odot \bm z;\bm \theta)}{ \bm q(\bm \theta)}, \tilde{\bm H}(\bm z;\bm \theta) = \frac{\bm H(\bm z;\bm \theta)}{\bm q(\bm \theta)}.  
\]
As discussed in \citep{Jagers2008PGFcondtioned}, $\tilde{\bm G}(\bm z;\bm \theta)$ itself represents a subcritical process. Let $\tilde{p}_{\bm H}^{(i)}(\bm d;\bm \theta) =p_{\bm H}^{(i)}(\bm d;\bm \theta) / q^{(i)}(\bm \theta) = [\bm z^{\bm d}]\tilde{H^{(i)}}(\bm z;\bm\theta)$ for $|\bm d|_1 <\infty$. Before discussing the qualitative relationship between  $\tilde{p}_{\bm H}^{(i)}(\bm d;\bm \theta)$ and $p_{\bm H}^{(i)}(\bm d;\bm \theta)$, first we see that  $\tilde{p}^{(i)}(\bm d;\bm \theta)$ is $C^s$ with respect to $\bm \theta$ within the supercritical region.  
\begin{rem}\label{rem:q_analytic}
Fix $\bm\theta_0\in\mathcal P^{\mathrm{sup}}$ and write $\bm q_0=\bm q(\bm\theta_0)\in(0,1)^n$. By Lemma~\ref{lem:q_by_criticality}, $\rho(\bm J_{\bm G}(\bm q_0;\bm\theta_0))<1$. To verify continuity of $\bm q$ at $\bm\theta_0$, set $\bm w=(\bm I-\bm J_{\bm G}(\bm q_0;\bm\theta_0))^{-1}\bm1\geq\bm1$. For each sufficiently small $t>0$, $\bm y=\bm q_0+t\bm w<\bm1$ satisfies
\[
 \bm G(\bm y;\bm\theta_0)-\bm y=-t\bm1+o(t)<\bm0 .
\]
For nearby $\bm\theta$ this strict inequality persists, so $\bm G^{\circ m}(\bm0;\bm\theta)\leq\bm q(\bm\theta)\leq\bm y$. The finite iterates are continuous at $\bm\theta_0$, and they increase to $\bm q_0$ there. Taking $m$ large and $t$ small proves continuity. The line-segment argument in Corollary~\ref{cor:theta_differentiability} now applies to $\bm q(\bm\theta)=\bm G(\bm q(\bm\theta);\bm\theta)$. In particular,
\begin{equation}
\frac{\partial \bm q(\bm\theta)}{\partial\theta_\ell}
=\left[\bm I-\bm J_{\bm G}(\bm q(\bm\theta);\bm\theta)\right]^{-1}
 (\partial_{\theta_\ell}\bm G)(\bm q(\bm\theta);\bm\theta), 
\label{eq:q_parameter_derivative}
\end{equation}
and it follows that $\bm q(\bm\theta)$ is also $C^s$. Therefore, from
\[
p_{\bm H}^{(i)}(\bm d;\bm \theta) = \tilde{p}_{\bm H}^{(i)}(\bm d;\bm \theta)  q^{(i)}(\bm \theta)  
\]
we see that $\tilde{p}_{\bm H}^{(i)}(\bm d;\bm \theta) $ is $C^s$ for $\bm \theta$. 
\end{rem}

Upon real-world application, suppose we have an observed dataset with independent clusters $\bm D=\{(i_\ell,\bm d_\ell):\ell=1,\ldots,N\}$. Then, the relationship between unconditional and conditional likelihoods
\begin{align}
L_{\mathrm{uncond}}(\bm\theta;\bm D)
&=\prod_{\ell=1}^N
p^{(i_\ell)}(\bm d_\ell;\bm\theta),
\label{eq:unconditional_likelihood_revised}\\
L_{\mathrm{cond}}(\bm\theta;\bm D)
&=\prod_{\ell=1}^N
\widetilde p^{(i_\ell)}(\bm d_\ell;\bm\theta), 
\label{eq:conditional_likelihood_revised}
\end{align}
can be written as
\begin{equation}
L_{\mathrm{uncond}}(\bm\theta;\bm D)
=\left(
\prod_{\ell=1}^N q^{(i_\ell)}(\bm\theta)
\right)
L_{\mathrm{cond}}(\bm\theta;\bm D).
\label{eq:unconditional_decomposition_revised}
\end{equation}
Note that Corollary~\ref{cor:theta_differentiability} guarantees that $L_{\mathrm{uncond}}(\bm\theta;\bm D)$ is $C^s$ across the criticality threshold, whereas $L_{\mathrm{cond}}(\bm\theta;\bm D)$ is also $C^s$ for $\bm \theta \in \mathcal P^{\mathrm{sup}}$ as in Remark \ref{rem:q_analytic}. In addition to the subcritical processes represented by $\tilde{\bm G}(\bm z;\bm \theta)$, the extra extinction-probability factor in Eq.~\eqref{eq:unconditional_decomposition_revised} gives
\[
 \mathrm d\log L_{\mathrm{uncond}}-\mathrm d\log L_{\mathrm{cond}}
 =\sum_{\ell=1}^N\mathrm d\log q^{(i_\ell)}(\bm\theta). 
\]
Let $N_i$ denote the number of clusters with initial type $i$. The two log likelihoods have the same slope in a given direction $\mathrm d\bm\theta$ if and only if
\[
 \sum_{i=1}^n N_i\,\mathrm d\log q^{(i)}(\bm\theta)
 = \sum_{i=1}^n N_i
   \sum_{j=1}^p \frac{\partial \log q^{(i)}(\bm\theta)}{\partial\theta_j}
   \,\mathrm d\theta_j \approx 0 .
\]
Though not necessarily decisive, in $L_{\mathrm{uncond}}(\bm\theta;\bm D)$ this trade-off may cause an extra geometry to $L_{\mathrm{cond}}(\bm\theta;\bm D)$ that might lead to different identification of parameters in statistical inference. Appendix~\ref{app:negative_multinomial} illustrates the sensitivity of extinction probability to the parameters that manifest the offspring that follows a negative multinomial distribution.

\section{Demonstration on Data}
\label{sec:demonstration}
This section illustrates the application of our framework on both a simulated dataset and an observed dataset with incomplete type information. As an example, we assume a two-type GWP with negative multinomial offspring distributions,
\begin{equation}
G^{(i)}(z_1,z_2;\bm\theta)
=\left(1+\frac{1}{k}\sum_{j=1}^2 K_{ij}(1-z_j)\right)^{-k},
\qquad i=1,2,
\label{eq:demo_nm_pgf}
\end{equation}
where $K_{ij}$ is the expected number of type-$j$ offspring from a type-$i$ parent, $k>0$ is a common dispersion parameter, and $R=\rho(\bm K)$. Note that the total progenys include the initial ancestor. We first compare numerical coefficients with a Lagrange--Good calculation and then illustrate likelihood evaluation across the criticality threshold. Next, we demonstrate our framework on MERS data. (details on numerical approaches are provided in Appendix \ref{app:dft_implement}, \ref{app:newton_picard_hybrid}, \ref{app:LG_validity}).

\subsection{Numerical comparison with the Lagrange-Good method}
\label{subsec:demo_sim_compare_LG}

First, to compare the total progeny PGF calculation given by the Lagrange--Good approach with that of ours, we design the underlying $2\times 2$ mean offspring matrix $\bm K= \left[ \frac{\partial G^{(i)}}{\partial z_j} \right]_{\bm z = \bm 1}$  using three scalars $(R,\alpha,p)$, where $R>0$, $\alpha\in(0,1)$, and $p\in(0,1)$, as follows: 
\begin{equation}
\bm B(\alpha,p)=
\begin{pmatrix}
\alpha p&\alpha(1-p)\\
(1-\alpha)(1-p)&(1-\alpha)p
\end{pmatrix},
\qquad
\bm K=R\frac{\bm B}{\rho(\bm B)}.
\label{eq:demo_B_K}
\end{equation}
Here, $\alpha$ controls the relative expected number of offspring generated by the two types, whereas $p$ controls within-type transmission. We used $R\in\{0.50,1.25\}$, $k\in\{0.1,1\}$, $\alpha=0.25$, and $p\in\{0.25,0.75\}$, and evaluated coefficients through degree $40$ in each type. Table~\ref{tab:comparison_K40_symmetric} in Appendix~\ref{app:LG_comparison} gives the numerical comparison. The largest reported absolute difference between the two calculations is $3.8\times10^{-15}$. The probability mass summed over this finite range can be smaller than the extinction probability because it omits larger finite clusters; this difference is not, by itself, a numerical error estimate.

\subsection{Likelihood surface for simulated cluster data}
\label{subsec:demo_likelihood_surface}
We now demonstrate the likelihood surface, with the aim of showing that its smoothness spans the criticality threshold. For this purpose, we simulate cluster data based on the contact matrix used in a previous study by Kucharski and Edmunds \citeyearpar{kucharski2015PCB}, where type 1 represents individuals aged $<20$ years and type 2 represents those aged $\geq20$ years. Let $M_{ij}$ denote the average number of contacts that a type-$j$ individual makes with type-$i$ individuals. We use
\begin{equation}
\bm M=\begin{pmatrix} 5.5& 1.3\\3.9& 4.6\end{pmatrix},
\qquad
\bm K(R)=R\left(\frac{\bm M}{\rho(\bm M)}\right)^{\!\top}.
\label{eq:demo_contact}
\end{equation}
Based on this matrix, we construct the negative multinomial PGF in Eq.~\eqref{eq:demo_nm_pgf}. We generated 20 independent clusters for each initial type, using $R=0.75$ and $k=0.1$, and evaluated the unconditional likelihood in Eq.~\eqref{eq:unconditional_likelihood_revised} without restricting the fitted process to be subcritical.

\begin{figure}[htbp]
\centering
\includegraphics[width=\textwidth]{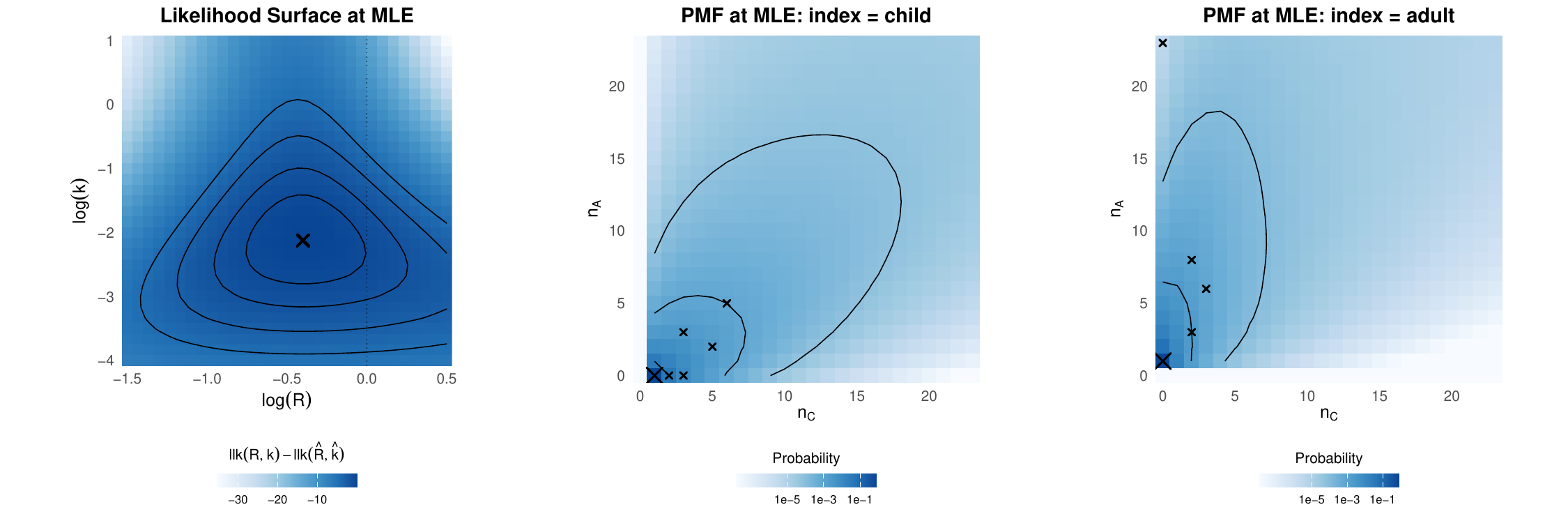}
\caption{Likelihood surface and estimated probability mass functions (PMFs) for the simulated two-type cluster data. Left: log likelihood relative to its maximum, as a function of $\log R$ and $\log k$. The black cross indicates the maximum likelihood estimate, and the vertical dotted line marks $R=1$. Middle and right: observed cluster sizes overlaid on the final-size PMFs for child and adult initial cases, respectively, evaluated at the maximum likelihood estimate. total progenys include the initial case.}
\label{fig:demo2_surface}
\end{figure}

The maximum likelihood estimates are $\widehat R=0.670$ and $\widehat k=0.120$ (Figure~\ref{fig:demo2_surface}). The likelihood is smoothly connected across $\log R=0$, consistently with Corollary~\ref{cor:theta_differentiability}. In the supercritical region, the contours extend toward larger $R$ and smaller $k$. The decomposition in Eq.~\eqref{eq:unconditional_decomposition_revised} provides a way to examine the contribution of extinction probabilities to this geometry. The shape of the full likelihood, however, also depends on the extinction-conditioned component; the descending contours alone do not establish that extinction probabilities are its sole cause.

\subsection{Likelihood construction with incomplete type information: MERS}
\label{subsec:mers}
\subsubsection{Dataset and methods}
In this section, we applied the framework to the Middle East respiratory syndrome (MERS) dataset in Cauchemez et al. \citeyearpar{cauchemez2014MERSLID}, which was later revisited in Kucharski and Edmunds \citeyearpar{kucharski2015PCB}. This dataset consists of 42 clusters involving 111 confirmed or probable human cases, with information on the age group ($<20$ vs. $\ge20$ years) and the index case for each cluster. For some transmission clusters, category information (children or adults) is lacking for cases that include the index cases. We consider a two-type GWP in which type 1 and type 2 represent individuals aged $<20$ years and $\ge20$ years, respectively, following the structure as in Kucharski and Edmunds \citeyearpar{kucharski2015PCB}.

For an infectious individual of type $i\in\{1,2\}$, let $(X_{i1},X_{i2})$ denote the numbers of secondary cases in the two age groups. The PGF of an offspring starting from type $i$ can be written as $G^{(i)}(z_{1},z_{2};\bm{\theta})=\mathbb{E}[z_{1}^{X_{i1}}z_{2}^{X_{i2}}], |z_{1}|<1, |z_{2}|<1$, where $\bm{\theta}$ are model parameters that will be fully described later. 
Following Kucharski and Edmunds \citeyearpar{kucharski2015PCB}, a next-generation matrix with the following structure was used:
\begin{equation}
 R\frac{\bm S \bm M}{\rho(\bm S \bm M)}=\lambda \begin{pmatrix} 1 & 0 \\ 0 & s \end{pmatrix} \begin{pmatrix} M_{11} & M_{12} \\ M_{21} & M_{22} \end{pmatrix} \left(=\bm K^\top \right),
\end{equation}
where $\bm S=\operatorname{diag}(1, s)$ expresses the relative susceptibility of adults compared with children ($s<1$ is assumed, as in \citep{kucharski2015PCB}), $\bm{M}=(M_{i j})$ is the contact matrix, and $\lambda$ is a scalar. 
Following this formulation, the negative multinomial PGF $G^{(i)}$ is defined as in Eq. \eqref{eq:demo_nm_pgf}. 

For an arbitrary cluster dataset $D=\{i,d_{1},d_{2}\}$ (where the index case is type $i$ and the numbers of type 1 and type 2 cases are $d_{1}$ and $d_{2}$, respectively), we may calculate $p_{i}(d_{1},d_{2};\bm \theta)=[z_{1}^{d_{1}}z_{2}^{d_{2}}]H^{(i)}(z_{1},z_{2};\bm{\theta})$. 

Let $I_D\in \{1, 2\}$ denote the type of the index case in cluster $D$. By considering $p(I_D =1)=\pi$, i.e., the probability that a reported cluster starts from an index case with type 1, and $p(I_D =2)=1-\pi$, the likelihood of observing the cluster can be simply written as $L_{D}=p(I_D = i)p_{i}(d_{1},d_{2};\bm{\theta}), i=1, 2$, for datasets with complete information on the type of both the index case and all cases. 

When complete information on the type of index case or type information for some cases is lacking, as for some clusters in the MERS dataset, some additional consideration is required.
A) The type of the index case is unknown, but complete information is available for the types of all cases:
By considering $\pi$, i.e., the probability that a reported cluster starts from an index case of type 1, 
the marginalized likelihood of observing the cluster is
\begin{equation}
L_{D}=\sum_{i=1}^{2}p(I_D = i)p^{(i)}(d_{1},d_{2};\bm{\theta})=\pi p^{(1)}(d_{1},d_{2};\bm \theta)+(1-\pi)p^{(2)}(d_{1},d_{2};\bm \theta).
\end{equation}
B) The type of the index case is known, but there is incomplete information on the types for all cases in the cluster:
When the type of the index case is known to be type $i$, the marginal likelihood of observing the cluster can be calculated as
\begin{equation}
L_{D}=p(I_D = i)\sum_{m=0}^{M}p^{(i)}(d_{1}^{\mathrm{obs}}+m,d_{2}^{\mathrm{obs}}+M-m;\bm \theta),
\end{equation}
where $M$ is the total number of cases in the cluster whose exact type is unknown (i.e., $M = N_{\mathrm{total}} - d_1^{\mathrm{obs}} - d_2^{\mathrm{obs}}$).

C) Both the type of the index case and complete information on all cases are missing:
By analogy to the preceding two cases, the marginalization is

{\small
\begin{equation}
L_{D}=\sum_{m=0}^{M}\left(\pi p^{(1)}(d_{1}^{\mathrm{obs}}+m,d_{2}^{\mathrm{obs}}+M-m;\bm \theta)+(1-\pi)p^{(2)}(d_{1}^{\mathrm{obs}}+m,d_{2}^{\mathrm{obs}}+M-m;\bm \theta)\right).
\end{equation}
}
Using calculations as in A--C when needed, the total likelihood for observing the entire dataset of MERS clusters is : 
\begin{equation}
L(\bm{\theta})=\prod_{D\in D_{\mathrm{MERS}}}L_{D}(\bm{\theta}),
\end{equation}
where $\bm{\theta}=\{\log R,\log k,\operatorname{logit} s,\operatorname{logit} \pi\}$ were used. For $\bm \theta$, weakly informative priors are introduced for statistical inference:
\begin{equation} 
\begin{aligned}
\log (R) \sim \mathrm{Normal}(0, 1^2), \\
\log (k) \sim \mathrm{Normal}(0, 2^2), \\
\operatorname{logit}(s) \sim \mathrm{Normal}(0, 5^2), \\
\operatorname{logit}(\pi) \sim \mathrm{Normal}(0, 5^2).
 \end{aligned}
 \end{equation}
We estimated the parameters by the Differential Evolution Markov Chain Monte Carlo (MCMC) algorithm (named ``DEzs'') in the R package ``BayesianTools'' \citep{Hartig2017-th} with four independent MCMC runs using three internal parallel subchains with 2,000 iterations. After discarding the first 1,000 iterations for each subchain, convergence was assessed and it was confirmed that all parameters satisfied $\widehat{R} <1.04$ and the effective sample size was $>690$. For convenience, we retained 1,200 posterior draws for every parameter to evaluate the model fit (additional diagnostic plots are provided in Figs. \ref{fig:mcmc_trace} -- \ref{fig:corr_plot}). 

\subsubsection{Results}
The parameter estimates are presented in Table \ref{tab:mers_posterior_summ}. The $R$ estimate is 0.764 (95\% CrI: 0.542, 1.137), which is comparable to previous estimates \citep{kucharski2015PCB}. We did not specifically impose a constraint on $R$ that forced the process to be subcritical, as can be observed in the posterior draws and uncertainty bound of $R$. The values of $k$,  $s$, and $\pi$ are estimated to be 0.438 (95\% CrI: 0.135, 4.092), 0.995 (95\% CrI: 0.856, 1.00),  and 0.048 (95\% CrI: 0.008, 0.138), respectively, where the estimated values of $s$ are in line with previous estimates \citep{kucharski2015PCB}. 
As shown in Fig. \ref{fig:corr_plot} in the Appendix \ref{app:mers}, the correlation between posterior samples suggests a ridge toward the direction of larger $\log R$ and smaller $\log k$ that extends into the supercritical domain. 

\begin{table}[t]
\centering
\caption{Parameters estimated from the MERS dataset.}
\label{tab:mers_posterior_summ}
\begin{tabular}{cc}
\hline
Parameter & Posterior Median (95\% CrI) \\
\hline
$R$ & 0.764 (0.542, 1.137)\\
$k$ & 0.438 (0.135, 4.092)\\
$s$ & 0.995 (0.856, 0.99999)\\
$\pi$ & 0.048 (0.008, 0.138)\\
\hline
\end{tabular}
\end{table}

Figure \ref{fig:mers_overlay}  confirms the fit of our inference results to the observed data by overlaying the observed clusters on the posterior mean probability mass matrix. For three of the 42 clusters, every compatible composition lay outside the 95\% highest-mass set in the corresponding panel. 

\begin{figure}[H]
    \centering
    \includegraphics[width=\linewidth]{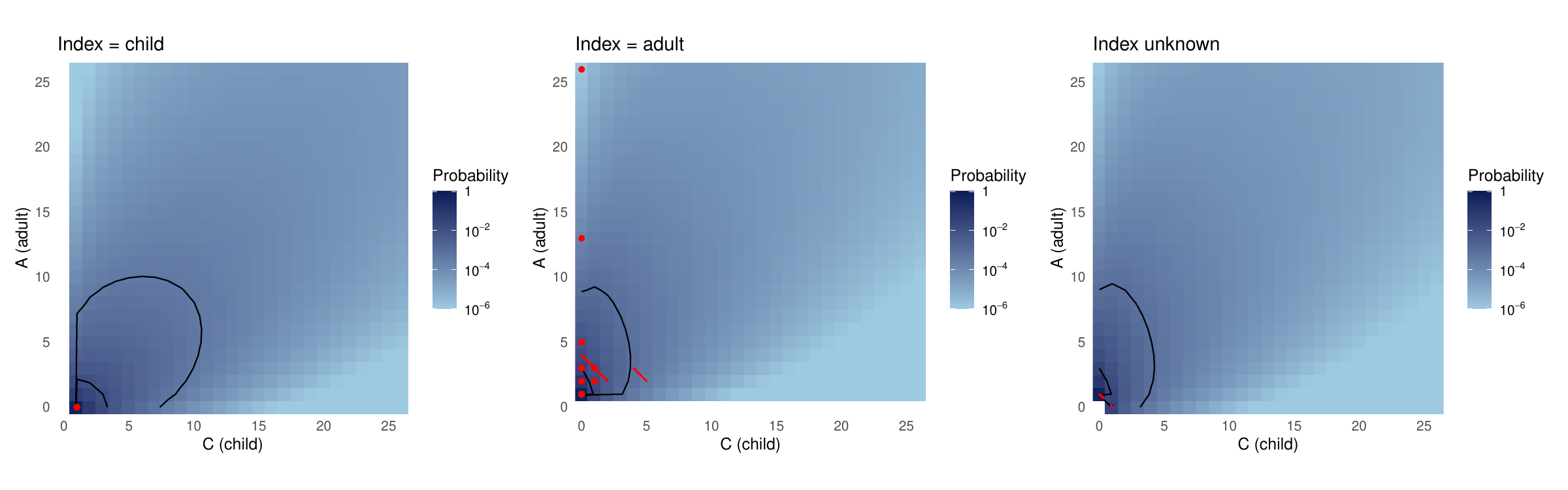}
    
    \caption{Posterior mean probability mass functions for child and adult cluster sizes under the fitted model. Black contours indicate the 50\%, 80\%, and 95\% highest-mass regions relative to the probability mass on the computed grid in each panel. Red points and line segments mark age compositions compatible with the observed clusters. Solid segments indicate incomplete composition with known index type, and dashed segments indicate unknown index type.}
    \label{fig:mers_overlay}
\end{figure}

\section{Conclusion}
We provided a foundation for directly computing the total progeny distribution of a multi-type GWP from its implicit PGF equation. Theorem \ref{thm:sol_region} shows that the finite final-size PGF $\bm H(\bm z; \bm \theta)$ as a solution to $\bm h=\bm z\odot\bm G(\bm h)$ is uniquely identified in $|\bm h|\leq\bm1$ for every $|\bm z|<\bm1$. Corollary \ref{cor:error_bound_picard} provides an explicit error bounds for Picard iterations that converges to $\bm H(\bm z; \bm \theta)$, whereas Lemma \ref{lem:contraction} describes the contraction of the mapping $\bm h \mapsto \bm z \odot \bm G(\bm h;\bm \theta)$ inside $\mathcal D_{\bm z}$ for processes without a final group.  These results provide a probabilistic basis for computing the total progeny distribution by obtaining the PGF values required for Cauchy coefficient extraction.

Additionally, Corollary \ref{cor:theta_differentiability} shows that, under certain assuptions the differentiability of the offspring distribution $\bm G(\bm z; \bm \theta)$ by $\bm \theta$ is inherited by $\bm H(\bm z; \bm \theta)$ and its coefficients, and consequently by the unconditional likelihood for a finite total progeny observation. Importantly, this is true even across the criticality threshold. Remark~\ref{rem:q_analytic} clarifies the differentiability of extinction probabilities and extinction-conditioned likelihoods in the supercritical region under the assumptions of Section \ref{subsec:identifiability}. The decomposition in Eq.~\eqref{eq:unconditional_decomposition_revised} describes the additional geometry of the unconditional likelihood in the strictly supercritical regime, which may cause trade-off between parameters and lead to identifiability issues upon stasitical inference.

Section \ref{sec:demonstration} demonstrates the alignment of results from our framework with those from the Lagrange--Good method,  visualization of the smoothness of unconditional likelihood across the criticality threshold, and applicability of our approach to real-world data with MERS cluster size data as an example. 

In summary, we studied the joint distribution of total progeny counts by types of population in multi-type GWPs based on probabilistic arguments, and revealed results that provide a basis for computing total progeny distributions in general. Our results are relevant to application of GWPs in various fields including, but not limited to, biology and epidemiology. 

\FloatBarrier
\backmatter

\section*{Statements and Declarations}

\bmhead{Funding}
\FSDFundingStatement

\bmhead{Competing interests}
There were no competing interests to declare which arose during the preparation or publication process of this article.

\bmhead{Ethics approval and consent to participate}
The MERS analysis used only previously published, publicly available data. No ethics approval was sought for this secondary analysis.

\bmhead{Consent for publication}
Not applicable.

\bmhead{Data and code availability}
The R scripts and input data are available at GitHub
(URL: \FSDGitHubURL) and Zenodo (URL: \FSDZenodoURL).

\bmhead{Materials availability}
Not applicable.
\bmhead{Author contributions}
\FSDAuthorContributions

\bmhead{Use of AI}
Since 2025, large language models (primarily GPT-5, GPT-6, and Gemini 3 families)
have been used for literature research, coding assistance, confirming internal consistency and screening for possible errors in the mathematical argument. Both authors reviewed the manuscript and take full responsibility for its content.

\setcounter{FSDmainfigures}{\value{figure}}
\setcounter{FSDmaintables}{\value{table}}
\begin{appendices}
\setcounter{figure}{\value{FSDmainfigures}}
\setcounter{table}{\value{FSDmaintables}}
\renewcommand{\thefigure}{\arabic{figure}}
\renewcommand{\thetable}{\arabic{table}}
\renewcommand{\theequation}{\thesection.\arabic{equation}}

\section{Standard facts for irreducible processes}
\label{app:standard_facts}
The supplementary discussions and proofs below use the assumptions in Section~\ref{sec:gwp_fs_intro}. Irreducibility is additionally assumed for the transience and criticality arguments.

\subsection{Transience of finite nonzero states}
\label{app:transience}
\begin{proof}[Proof sketch of Lemma~\ref{lem:transience}]
Following Theorem 6.1 and related results in Harris~\citeyearpar[Chapter~II, Theorem~6.1]{harris1963theory}, a brief extension for an irreducible and non-singular process will be introduced below. 

Let $\bm \eta$ be a finite nonzero initial state, $\mathcal S=\{j:q^{(j)}=0\}$, and $\zeta(\bm x)=\sum_{j\in\mathcal S}x_j$. If $\mathcal S=\varnothing$, extinction has positive probability from $\bm \eta$, so the first case of Harris's proof applies. Otherwise, $\zeta (\bm Z_m)$ is nondecreasing in $m$ with probability $1$ by Harris~\citeyearpar[Chapter~II, Lemma~6.1]{harris1963theory}, which does not require positive regularity. Irreducibility ensures that any type of offspring can happen with positive probability at some finite generation from the nonzero initial state $\bm \eta$.

If $\mathcal S$ contains all types $1, 2, \ldots, n$, from nonsingularity assumptions, there exists at least one type $i$ such that at least two offspring happen with positive probability. Once the process reaches this type, after one further generation we see that $\Pr\left(\zeta (\bm Z_m)>\zeta (\bm \eta) \mid  \bm Z_{0}=\bm \eta \right)>0$ for some finite $m$.

For the case that $\mathcal S$ contains $r$ types ($0<r<n$),  let $j\notin\mathcal S$, and $\bm v$ be an arbitrary reachable state at the $\ell$-th generation with $v_j>0$. Then, by irreducibility, there are $a\in\mathcal S$ and $p\geq1$ such that the probability of a type-$j$ parent having at least one type-$a$ offspring after $p$ generations is positive. Conditional on $\bm Z_\ell=\bm v$, we may consider splitting $\bm Z_{\ell+p}$ into $\bm Z'_{\ell+p}$ and $\bm Z''_{\ell+p}$, that are descendants of individuals in $\bm v$ belonging to $\mathcal S$ and its complement, respectively. 
Then
\[
\zeta(\bm Z'_{\ell+p})\geq \zeta(\bm v) \geq \zeta(\bm \eta), \qquad
\Pr 
\left(
\zeta(\bm Z^{''}_{\ell+p})>0 \mid \bm Z_\ell=\bm v\right)>0.
\]
Since $\Pr(\bm Z_\ell=\bm v \mid \bm Z_0=\bm\eta)>0$, this gives
\[
\Pr\left(
\zeta(\bm Z_{\ell+p}) = 
\zeta(\bm Z'_{\ell+p}) + \zeta(\bm Z^{''}_{\ell+p})
>\zeta(\bm\eta)
\mid\bm Z_0=\bm\eta
\right)>0.
\]
In either nonempty case, $\bm \eta$ is transient. 
\end{proof}

\subsection{Minimal fixed point and the criticality classification}
\label{app:criticality}

\begin{proof}[Proof of the irreducible, nonsingular case of Lemma ~\ref{lem:q_by_criticality}]
Since the iteration $\bm G^{\circ m}(\bm0)$ represents extinction at or earlier than $m$-th generation component-wise, all components are nondecreasing in $m$ and are upper-bounded by $\bm 1$. Therefore, for each origin type $i$ there exists a limit $\left(\bm G^{\circ m}(\bm0) \right)^{(i)}
=\Pr(\tau^{(i)}\leq m)\to q^{(i)}$, which satisfies $\bm q=\bm G(\bm q)$. For any fixed point $\bm a\in[0,1]^n$, monotonicity gives $\bm G^{\circ m}(\bm0)\leq\bm a$ and $\bm q\leq\bm a$. 

Hereafter we assume that the process is irreducible and nonsingular. Then, let $\bm v_{\bm K}>\bm0$ be the right Perron vector of $\bm K$. In the subcritical and critical cases, for a process started from type $i$, the following holds for the expected value of $W_m^{(i)}:=(\bm Z_m^{(i)})^\top\bm v_{\bm K}$: 
\[
\mathbb E_i[W_m^{(i)}]
=(\bm K^m\bm v_{\bm K})_i
=\rho(\bm K)^m(v_{\bm K})_i
\leq(v_{\bm K})_i. 
\]
By Lemma \ref{lem:transience} either $|\bm Z_m^{(i)}|_1 \to 0  \text{ or } \infty$ at the limit $m \to \infty$. As the existence of any positive probability mass corresponding to $|\bm Z_m^{(i)}|_1 \to \infty$ contradicts this inequality, $|\bm Z_m^{(i)}|_1 \to 0$, or, $q^{(i)} = \Pr_{i}(|\bm Z_{\infty}^{(i)}|_1 = 0) = 1 $. Regarding uniqueness, if there is another fixed point $\bm a\in[0,1]^n$, $\bm G^{\circ m}(\bm 0) \leq \bm G^{\circ m}(\bm a) \leq \bm 1$ leads to $\bm G^{\circ m }(\bm a) \to \bm 1 = \bm q$ as $m \to \infty$ for subcritical and critical processes. 

For supercritical processes, fix $\epsilon>0$ sufficiently small that $\bm0<\bm1-\epsilon\bm v_{\bm K}$ and 
\[
\bm G(\bm1-\epsilon\bm v_{\bm K})
=\bm1-\epsilon\rho(\bm K)\bm v_{\bm K}+o(\epsilon)
<\bm1-\epsilon\bm v_{\bm K}.
\]
By monotonicity, induction gives $\bm G^{\circ m}(\bm0)\leq\bm1-\epsilon\bm v_{\bm K}$ for every $m$.  Therefore, $\bm q = \lim_{m \to \infty}\bm G^{\circ m}(\bm0 )<\bm1$. Next, both $\bm q$ and $\bm1$ are fixed points, so
\[
\begin{split}
\bm1-\bm q
&=\bm G(\bm1)-\bm G(\bm q)\\
&=\left(\int_0^1\bm J_{\bm G}
 (\bm q+t(\bm1-\bm q)) \,\mathrm{d}t\right)(\bm1-\bm q). 
\end{split}
\]
Because $\bm B=\int_0^1\bm J_{\bm G}
 (\bm q+t(\bm1-\bm q)) \,\mathrm{d}t$ is irreducible, $\rho(\bm B)=1$. Also, $\bm J_{\bm G}(\bm q) \leq \bm B$ holds with strict inequality for at least one component. Therefore, $\rho(\bm J_{\bm G}(\bm q)) < \rho(\bm B)=1$. 

Again, regarding uniqueness, suppose there is another fixed point $\bm b \in [0, 1)^n$. Then, by induction we see that  $\bm G^{\circ m}(\bm0) \leq \bm G^{\circ m}(\bm b) = \bm b$, and $\bm q\leq\bm b$ at the limit $m \to \infty$. Also, $\rho(\bm J_{\bm G}(\bm b))<1$ follows from the same argument. Due to monotonicity of $\bm J_{\bm G}$ as a function of positive real vectors, 
\[
\begin{split}
\bm b-\bm q
&=\left(\int_0^1\bm J_{\bm G}
 (\bm q+t(\bm b-\bm q))\,\mathrm{d}t\right)(\bm b-\bm q)\\
&\leq \bm J_{\bm G}(\bm b)(\bm b-\bm q), 
\end{split}
\]
which leads to  $\bm q = \bm b$ due to similar discussion as above because $\rho(\bm J_{\bm G}(\bm b))<1$.  
\end{proof}

\section{Alternative proof of the sharper contraction}
\label{app:contraction_H}
The finite total progeny PGF gives another choice of weights in Lemma~\ref{lem:contraction}.
\begin{lemma}
\label{lem:jacobian_weights}
Let $\bm A=\bm J_{\bm G}(\bm r\odot\bm q)$, where $\bm0 <\bm r<\bm1$. There exists a vector $\bm w>\bm0$ such that $\bm A\bm w\leq\bm w$, with strict inequality for $i\in\mathcal S^c$.
\end{lemma}
\begin{proof}
There exists $s \in (0, 1)$ such that $\bm r \odot \bm q \leq \bm H(s\bm1) \leq \bm q$, where equality holds for $\mathcal S$-components. For $\bm v = \bm H(\bm1)-\bm H(s\bm1) = \bm q-\bm H(s\bm1)$, integrating along the line from $\bm z = \bm H(s\bm1)$ to $\bm z = \bm H(\bm1) = \bm q$ gives
\[
\bm v \geq \bm q - \frac{\bm H(s\bm1)}{s} \geq \bm G(\bm H(\bm 1))-\bm G(\bm H(s\bm1)) \geq \bm J_{\bm G}((\bm H(s\bm1)))(\bm q-\bm H(s\bm1))\geq \bm A\bm v
\]
The first inequality is strict for $i\in\mathcal S^c$ because $H^{(i)}(s\bm1)>0$, whereas $\bm v_{\mathcal S}=\bm0_{\mathcal S}$. Therefore, taking $\bm w_{\mathcal S}=\epsilon \bm 1_{\mathcal S}$ and  $\bm w_{\mathcal S^c}=\bm v_{\mathcal S^c}$ for a sufficiently small $\epsilon>0$ suffices to see contraction as in Lemma \ref{lem:contraction}.
\end{proof}

\section{Implicit differentiation of extinction probabilities by parameters of the negative multinomial distributions}
\label{app:negative_multinomial}

Here details of the general description in  Eq. \eqref{eq:q_parameter_derivative} are provided for special cases with negative multinomial distributions with offspring PGFs
\[
G^{(i)}(\bm z)
=\left(1+\frac{1}{k_i}\sum_{m=1}^nK_{im}(1-z_m)\right)^{-k_i},
\qquad k_i>0,
\]
where the dispersion may depend on the parent's type. Note that an irreducible, supercritical process is assumed with $\bm q=\bm G(\bm q)$ where $\bm q\in(0,1)^n$. 

\subsection{Differentiation by elements of the mean offspring matrix}
A small perturbation of the mean offspring matrix $\bm K= \left[ \frac{\partial G^{(i)}}{\partial z_j}\right]_{\bm z = \bm 1}$ causes a small change in $\bm q$. For simplicity, we consider a perturbation of an arbitrary $\ell$th column of $\bm K$. For fixed $\ell$, the matrix
$[\partial q_i/\partial K_{j\ell}]_{i,j}$ is an $n\times n$ slice of the full $n\times n\times n$ sensitivity array.  We denote this perturbation to $\bm{K_{\cdot \ell}} = (K_{1\ell}, \ldots, K_{n\ell})^\top$, the $\ell$th column of $\bm K$, as $\bm{\Delta}_{\bm{K_{\cdot \ell}}} = (\Delta_{K_{1 \ell}}, \ldots, \Delta_{K_{n\ell}})^\top$. Taking the directional derivative of $\bm q=\bm G(\bm q)$ in this column direction gives
\begin{equation}
 \left[  \frac{\partial q^{(i)}}{\partial K_{j\ell}}  \right]\bm{\Delta}_{\bm{K_{\cdot \ell}}}
 = \left( \left[  \frac{\partial G^{(i)}}{\partial K_{j\ell}} \right]_{\bm{z}=\bm{q}} 
 + \left[ \frac{\partial G^{(i)}}{\partial z_j} \right]_{\bm{z}=\bm{q}} \left[  \frac{\partial q^{(i)}}{\partial K_{j\ell}} \right] \right) \bm{\Delta}_{\bm{K_{\cdot \ell}}}, 
\end{equation}
which leads to 
\begin{equation}
\left( \bm{I} -  \left[ \frac{\partial G^{(i)}}{\partial z_j} \right]_{\bm{z}=\bm{q}} \right)  \left[   \frac{\partial q_i}{\partial K_{j\ell}}  \right]
 = \left[ \frac{\partial G^{(i)}}{\partial K_{j\ell}} \right]_{\bm{z}=\bm{q}} . 
\end{equation}
For $ \bm{J_G}(\bm{q}) =\left[ \frac{\partial G^{(i)}}{\partial z_j} \right]_{\bm{z}=\bm{q}}$, we know that $\rho( \bm{J_G}(\bm q))<1$. Therefore, 
\begin{equation}
\left[  \frac{\partial q^{(i)}}{\partial K_{j\ell}} \right]
 = \left( \sum_{m=0}^\infty {\bm{J_G}(\bm{q})}^m  \right) \left[  \frac{\partial G^{(i)}}{\partial K_{j\ell}} \right]_{\bm{z}=\bm{q}}. 
\end{equation}
On the right-hand side, $\sum_{m=0}^\infty {\bm{J_G}(\bm{q})}^m$ is a strictly positive matrix because $\bm{J_G}(\bm{q})$ is irreducible. For $\left[  \frac{\partial G^{(i)}}{\partial K_{j\ell}} \right]_{\bm{z}=\bm{q}}$,  
\begin{equation}
\left. \frac{\partial G^{(i)}}{\partial  K_{j\ell}} \right|_{\bm{z}=\bm{q}} =\begin{cases}
 - \left(1 -\ q^{(\ell)} \right) \left(q^{(i)} \right)^{1+1/k_i}<0  &  (i=j);\\
 0  &  (i \neq j).
\end{cases}
\end{equation}
In summary, $\left[ \frac{\partial q^{(i)}}{\partial K_{j\ell}}  \right]$ is a product of a positive matrix and a negative diagonal matrix. Therefore, it follows that $\frac{\partial q^{(i)}}{\partial K_{j\ell}}<0 $.

\subsection{Differentiation by overdispersion parameters}
Similarly, taking the directional derivative of $\bm q=\bm G(\bm q)$ for a small perturbation vector $\bm{\Delta}_{\bm{k}} = (\Delta_{k_1}, \ldots, \Delta_{k_n})^\top$ gives
\begin{equation}
 \left[ \frac{\partial q^{(i)}}{\partial k_j}  \right]\bm{\Delta}_{\bm{k}} 
 = \left( \left[  \frac{\partial G^{(i)}}{\partial k_j} \right]_{\bm{z}=\bm{q}} 
 + \left[ \frac{\partial G^{(i)}}{\partial z_j}  \right]_{\bm{z}=\bm{q}}\left[  \frac{\partial q^{(i)}}{\partial k_j} \right] \right) \bm{\Delta}_{\bm{k}}, 
\end{equation}
which leads to 
\begin{equation}
\left( \bm{I} -  \left[ \frac{\partial G^{(i)}}{\partial z_j}  \right]_{\bm{z}=\bm{q}} \right) \left[ \frac{\partial q^{(i)}}{\partial k_j}  \right]  
 = \left[  \frac{\partial G^{(i)}}{\partial k_j} \right]_{\bm{z}=\bm{q}}.  
\end{equation}
This can be rearranged as 
\begin{equation}
\left[  \frac{\partial q^{(i)}}{\partial k_j} \right] 
 = \left( \sum_{m=0}^\infty {\bm{J_G}(\bm{q})}^m  \right)\left[ \frac{\partial G^{(i)}}{\partial k_j} \right]_{\bm{z}=\bm{q}}. 
\end{equation}
Again, $\sum_{m=0}^\infty {\bm{J_G}(\bm{q})}^m$ is a strictly positive matrix because $ \bm{J_G}(\bm{q})$ is irreducible. For $\left[ \frac{\partial G^{(i)}}{\partial k_j} \right]_{\bm{z}=\bm{q}}$, $\frac{\partial G^{(i)}}{\partial k_j} =0$ when $i \neq j$. For $i = j$, 
\begin{equation}
\begin{split}
\left. \frac{\partial \log G^{(i)}}{\partial k_i} \right|_{\bm{z}=\bm{q}} 
&= -\log\left( 1 + \frac{1}{k_i} \sum_{m=1}^n K_{im}\left( 1 -  q^{(m)} \right) \right) + \frac{\frac{1}{k_i} \sum_{m=1}^n  K_{im}\left( 1 -  q^{(m)} \right) }{ 1 + \frac{1}{k_i} \sum_{m=1}^n K_{im}\left( 1 -  q^{(m)} \right)} \\
&=\log\left(\left(q^{(i)}\right)^{1/k_i}\right) + 1 - \left(q^{(i)}\right)^{1/k_i}.
\end{split}
\end{equation}
The final line is strictly negative ($f(x) = 1 - x + \log(x)$ is a strictly increasing function of $x \in (0, 1)$ and $\lim_{x \to 1-0} f(x) = 0$). Therefore,  $\left[  \frac{\partial G^{(i)}}{\partial k_j} \right]_{\bm{z}=\bm{q}}$ is a matrix with negative diagonal elements, and it follows that $\frac{\partial q^{(i)}}{\partial k_j} < 0$. 

\section{Numerical implementation}
\subsection{Implementation of Cauchy Coefficient formula}
\label{app:dft_implement}
The calculation of PGF coefficients using Cauchy integrals has previously been described in an epidemiological context \citep{Miller2018PGFprimer, Roberts2023EthicalControl_sciadv}, which can be generalized for multivariate cases as 
\begin{equation}
p^{(i)}(d_{1}, \ldots, d_{n})= \frac{1}{(2\pi \mathrm{i})^{n}} \oint_{\Gamma_{\bm{z}}} \frac{H^{(i)}(z_{1}, \ldots, z_{n})}{z_{1}^{d_{1}+1} \cdots z_{n}^{d_{n}+1}} \,\mathrm{d}z_{1}\cdots\mathrm{d}z_{n},
\end{equation}
which we evaluate on multidimensional grids along the contour $\Gamma_{\bm{z}}=C_{z_{1}} \times \cdots \times C_{z_{n}}$ using the discrete Fourier transform (DFT) \citep{trefethen2014trapezoidal}. 
Here, each $C_{z_{j}}$ is a circle defined by $|z_j| = r_j$, with $0<r_j<1$. Remark~\ref{rem:coefficient_error} bounds the errors from the approximate PGF values and the finite grid.

Abiding by general recommendations to minimize rounding and aliasing errors \citep{bornemann2011cauchy, AbateWhitt1992QS}, we construct the contour as follows:
\begin{itemize}
    \item Subcritical or critical: $\bm{r} = (0.95, \ldots, 0.95)$.
    \item Supercritical: find a vector $\bm{r}^{\#}=\bm{q}+t^{\#}(\bm{1}-\bm{q})$ with $t^{\#}\in [0,1]$ that minimizes $|\rho(\bm{J}_{\bm{G}}(\bm{r}^{\#}))-0.975|$. Then choose $r_j = \min(r_j^{\#}, 0.95)$.
\end{itemize}
After choosing $\Gamma_{\bm{z}}$, in the main text, we evaluated the integrand of the Cauchy coefficient formula on multidimensional grids on $\Gamma_{\bm{z}}$ using DFT. As an example, for the two-dimensional cases we consider later, the grids are 
\begin{equation}
\bm{z}(l_{1},l_{2};r_1, r_2)=(r_{1}e^{2\pi \mathrm{i}\frac{l_{1}}{N}},r_{2}e^{2\pi \mathrm{i}\frac{l_{2}}{N}}), \quad l_{1}, l_{2}=0,\ldots,N-1,
\end{equation}
and the Cauchy integral can be approximated as 
\begin{align}
\label{eq:coef_dft}
p^{(i)}(d_{1},d_{2}) &\approx \frac{1}{N^{2}}\sum_{l_{1},l_{2}=0}^{N-1}\frac{H^{(i)}(\bm{z}(l_{1},l_{2};r_1, r_2))}{r_{1}^{d_{1}}r_{2}^{d_{2}}}\exp\left(-\frac{2\pi \mathrm{i}}{N}(d_{1}l_{1}+d_{2}l_{2})\right).
\end{align}
In practice, we set the number of grids to $N=\max(2^5, 2^{(\lceil \log_2(d_{\max}) \rceil +2)})$, where $d_{\max}$ is the maximum degree in the dataset; this is a simple rule aligning with general recommendations to reduce aliasing errors while maintaining computational efficiency \citep{trefethen2014trapezoidal, cooley1965fft}.

\subsection{Numerical solving of the total progeny equation}
\label{app:newton_picard_hybrid}
Besides numerical implementation of Picard iteration for solving the total progeny equation, solving
\begin{equation}
\bm{F}_{\bm{z}}(\bm{h})=\bm{h}- \bm{z} \odot \bm{G}(\bm{h})=\bm{0}, \end{equation}
or, equivalently, a minimization problem of $\psi_{\bm{z}}(\bm{h})= \| \bm{F}_{\bm{z}}(\bm{h}) \|^{2}_{2} $ by Newton's method with backtracking \citep{nocedal2006Optimization} is a candidate to solve this system. As the Jacobian matrix $\bm{J}_{\bm{F}_{\bm{z}}}(\bm{h})$ is invertible within the region $\bm{h} \in \mathcal D_{\bm z} $ due to Lemma \ref{lem:contraction} except that the underlying process has a final group, a directional derivative in the Newton direction $\bm{\Delta}(\bm{h}) = -\bm{J}_{\bm{F}_{\bm{z}}}^{-1}(\bm{h})\bm{F}_{\bm{z}}(\bm{h})$ satisfies 
\begin{equation}
\frac{\mathrm{d}}{\mathrm{d}\alpha}\psi_{\bm{z}}(\bm{h}+\alpha\bm{\Delta}(\bm{h}))|_{\alpha=0}=-2\psi_{\bm{z}}(\bm{h}).
\end{equation}
Whenever $\psi_{\bm z}(\bm h)>0$, this directional derivative is negative. Starting from $\bm h^{(0)}\in  \mathcal D_{\bm z}$, we may accept Newton steps of the form
\begin{equation}
\bm h^{(k+1)}=\bm h^{(k)}+\alpha^{(k)}\bm\Delta(\bm h^{(k)}),
\qquad 0<\alpha^{(k)}\leq1,
\end{equation}
provided that $\bm h^{(k+1)}\in \mathcal D_{\bm z}$ and $\psi_{\bm z}(\bm h^{(k+1)})<\psi_{\bm z}(\bm h^{(k)})$ hold. If no such step is accepted, a Picard step remains available in $ \mathcal D_{\bm z}$ by Lemma~\ref{lem:contraction}. Though global convergence is not the focus here, Newton steps as above with a fallback to a Picard step whenever needed should work. This approach, together with the computation of Cauchy integrals as described in Appendix \ref{app:dft_implement}  was used in Section~\ref{sec:demonstration}. 

\section{Numerical comparison with the Lagrange--Good method}
\subsection{Numerical approach for applying the Lagrange--Good formula}
\label{app:LG_validity}
For an irreducible negative multinomial offspring as in Eq. \eqref{eq:demo_nm_pgf}, let 
\begin{equation}
\label{eq:Gi_NB_form}
G^{(i)}(\bm h)
=\mathrm{denom}_i(\bm h)^{-k},
\qquad
\mathrm{denom}_i(\bm h)
=1+\frac{1}{k}\sum_{j=1}^n K_{ij}\left(1- h_j\right).
\end{equation}
The Jacobian matrix in the Lagrange--Good Cauchy integral is $
\bm J_{\bm z \to \bm h} (\bm h)\;=\; \left[ \frac{\partial z_i}{\partial h_j} \right]$, where
\begin{equation}
\label{eq:J_factorization}
\frac{\partial z_i}{\partial h_j}
=\frac{1}{G^{(i)}(\bm h)}
\left(
\delta_{ij} - \frac{h_i}{G^{(i)}(\bm h)}\frac{\partial G^{(i)}(\bm h)}{\partial h_j}
\right).
\end{equation}
Thus, $\bm J_{\bm z \to \bm h} (\bm h)$ can be factorized as
\begin{equation}
\label{eq:J_core}
\bm J_{\bm z \to \bm h} (\bm h) = \operatorname{diag}\!\left(\frac{1}{G^{(1)}(\bm h)},\ldots,\frac{1}{G^{(n)}(\bm h)}\right)
\bm A(\bm h),
\qquad
A_{ij}(\bm h) \;=\; \delta_{ij} - \frac{h_i}{G^{(i)}(\bm h)}\frac{\partial G^{(i)}(\bm h)}{\partial h_j}.
\end{equation}
In this case, for $\bm h$ such that $|\bm h| < \bm q$, it follows that for $\bm B(\bm h) = \left[\frac{h_i}{G^{(i)}(\bm h)}\frac{\partial G^{(i)}(\bm h)}{\partial h_j} \right]_{ij}$ , 
\[
|\bm B| \leq \operatorname{diag}{\left( \frac{|\bm h|}{\bm q}\right)}\bm J_{\bm G}(\bm q). 
\]
Since $|\frac{\bm h}{\bm q}|<\bm 1$ and $\rho(\bm J_{\bm G}(\bm q))\le1$ by Lemma~\ref{lem:jacobian_at_extinction}, we have $\rho(\bm B(\bm h))<1$ for $|\bm h|<\bm q$. Therefore, $\bm A(\bm h)$ is invertible inside $|\bm h|< \bm q$. With a two-type case as an example, we selected the polyradii $\bm r=(r_1,r_2)$ of  
\[
\Gamma_{\bm h}(\bm r)  = C_{h_1}(r_1)\times C_{h_2}(r_2),
\qquad
h_\ell=r_\ell e^{\mathrm{i}\theta_\ell},\ \theta_\ell\in[0,2\pi),
\]
as 
\begin{equation}
\label{eq:rad_init}
r_{\ell}=\min\{0.99q_\ell,\ r_{\max}\},
\qquad r_{\max}=0.95,
\end{equation}
where $\bm q=(q_1,q_2)$ is the extinction probability vector in this case.

\subsection{Results of the comparison with LG method}
\label{app:LG_comparison}
Table~\ref{tab:comparison_K40_symmetric} reports the calculations described in Section~\ref{subsec:demo_sim_compare_LG}. The finite coefficient range has maximum degree $40$ in each type and includes the initial ancestor. The tabulated probability sums therefore omit larger finite total progenys.
\begin{table}[htbp]
\centering
\caption{Numerical comparison between the proposed method and the Lagrange--Good (LG) implementation (max degree: $40$ for each type).}
\label{tab:comparison_K40_symmetric}
\small
\setlength{\tabcolsep}{4pt}
\begin{tabular}{ccccccccc}
\hline
$R$ & $k$ & $\alpha$ & $p$ & index $i$ & $q^{(i)}$ & sum(LG) & sum(direct) & max. abs. diff. \\
\hline
0.50 & 0.1 & 0.25 & 0.25 & 1 & 1.0000 & 0.9993 & 0.9993 & $2.098 \times 10^{-16}$ \\
 &  &  &  & 2 & 1.0000 & 0.9978 & 0.9978 & $1.817 \times 10^{-16}$ \\
0.50 & 0.1 & 0.25 & 0.75 & 1 & 1.0000 & 0.9998 & 0.9998 & $4.897 \times 10^{-16}$ \\
 &  &  &  & 2 & 1.0000 & 0.9971 & 0.9971 & $3.382 \times 10^{-16}$ \\
0.50 & 1 & 0.25 & 0.25 & 1 & 1.0000 & 1.0000 & 1.0000 & $1.11 \times 10^{-16}$ \\
 &  &  &  & 2 & 1.0000 & 1.0000 & 1.0000 & $1.599 \times 10^{-16}$ \\
0.50 & 1 & 0.25 & 0.75 & 1 & 1.0000 & 1.0000 & 1.0000 & $2.229 \times 10^{-16}$ \\
 &  &  &  & 2 & 1.0000 & 0.9999 & 0.9999 & $1.617 \times 10^{-16}$ \\
1.25 & 0.1 & 0.25 & 0.25 & 1 & 0.9749 & 0.9655 & 0.9655 & $1.218 \times 10^{-16}$ \\
 &  &  &  & 2 & 0.9521 & 0.9338 & 0.9338 & $1.11 \times 10^{-16}$ \\
1.25 & 0.1 & 0.25 & 0.75 & 1 & 0.9917 & 0.9870 & 0.9870 & $2.551 \times 10^{-16}$ \\
 &  &  &  & 2 & 0.9589 & 0.9389 & 0.9389 & $2.217 \times 10^{-16}$ \\
1.25 & 1 & 0.25 & 0.25 & 1 & 0.8705 & 0.8647 & 0.8647 & $3.679 \times 10^{-15}$ \\
 &  &  &  & 2 & 0.7540 & 0.7431 & 0.7431 & $3.794 \times 10^{-15}$ \\
1.25 & 1 & 0.25 & 0.75 & 1 & 0.9568 & 0.9539 & 0.9539 & $1.552 \times 10^{-15}$ \\
 &  &  &  & 2 & 0.7878 & 0.7736 & 0.7736 & $1.002 \times 10^{-15}$ \\
\hline
\end{tabular}
\end{table}

\section{Results and diagnostics from MERS data analysis}
\label{app:mers}

Figure \ref{fig:mcmc_trace} displays the trace plots and density estimates, Figure \ref{fig:prior_posterior} displays densities of prior and posterior distributions, and Figure \ref{fig:corr_plot} displays the correlation between model parameters. In all figures, parameters 1, 2, 3, 4 denote $\log R$, $\log k$, $\operatorname{logit}(s)$, $\operatorname{logit}(\pi)$, respectively.

\begin{figure}[htbp]
    \centering
    \includegraphics[width=\linewidth]{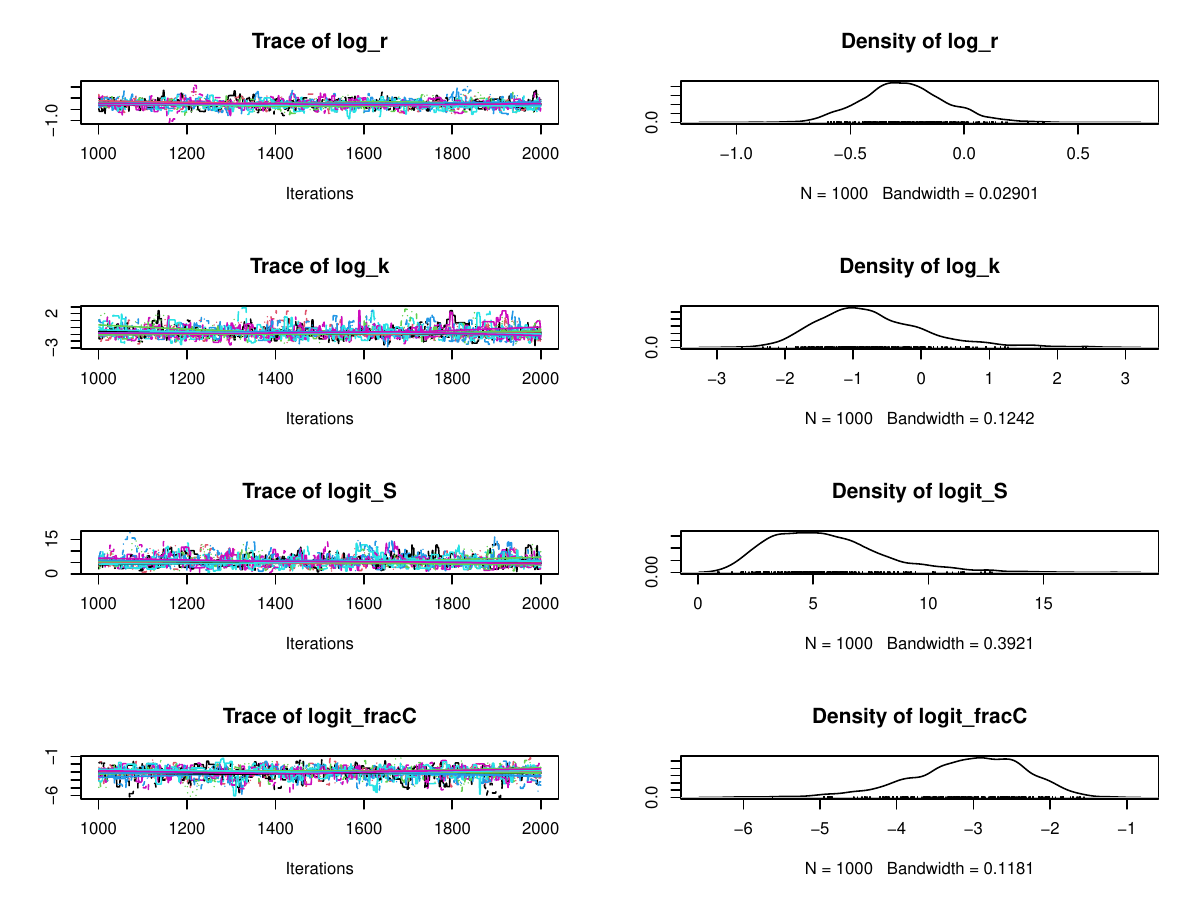}
    
    \caption{MCMC trace plots (left column) and marginal posterior density estimates (right column) for the model parameters.}
    \label{fig:mcmc_trace}
\end{figure}

\begin{figure}[htbp]
    \centering
    \includegraphics[width=\linewidth]{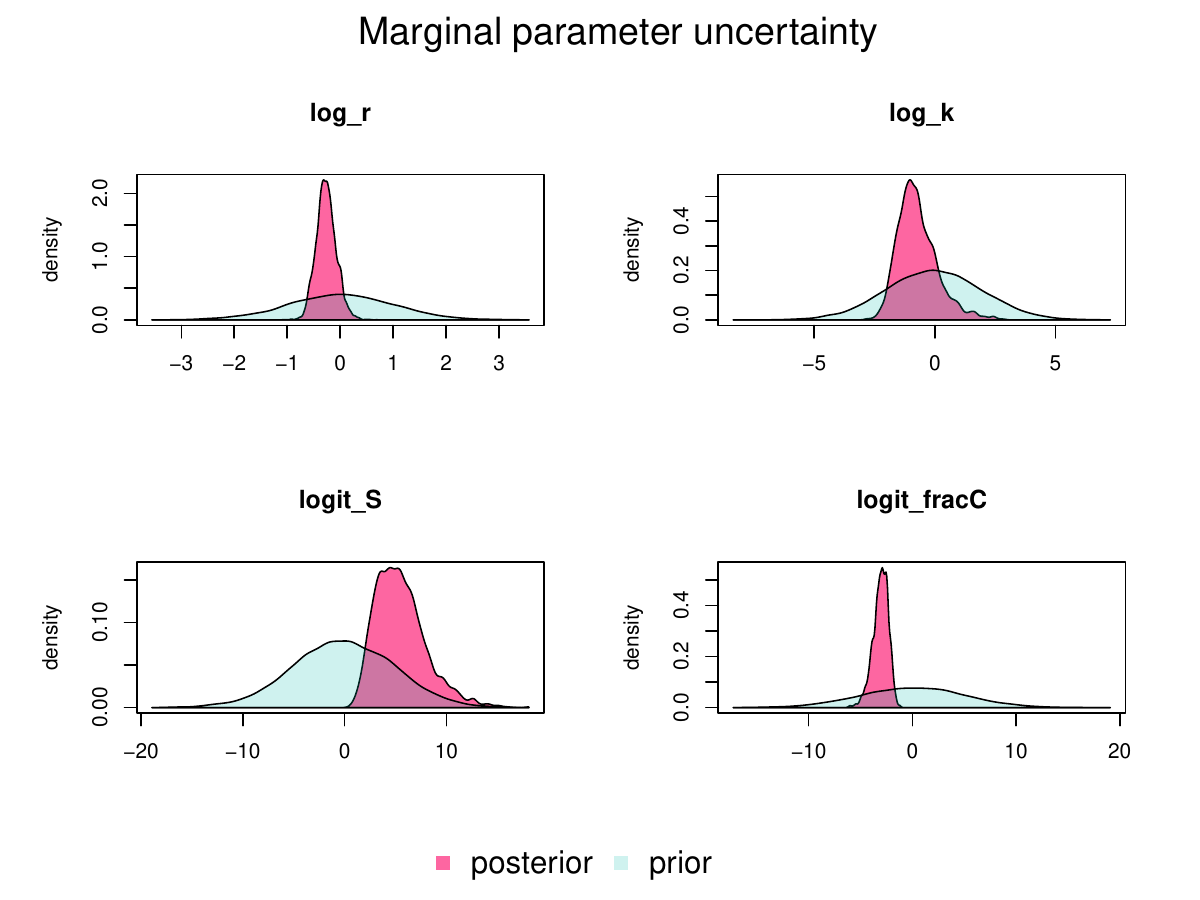}
    
    \caption{Overlay of prior and posterior densities for the model parameters. The light-blue shaded areas represent the prior distributions, whereas the red shaded areas represent the posterior distributions.}
    \label{fig:prior_posterior}
\end{figure}

\begin{figure}[htbp]
    \centering
    \includegraphics[width=\linewidth]{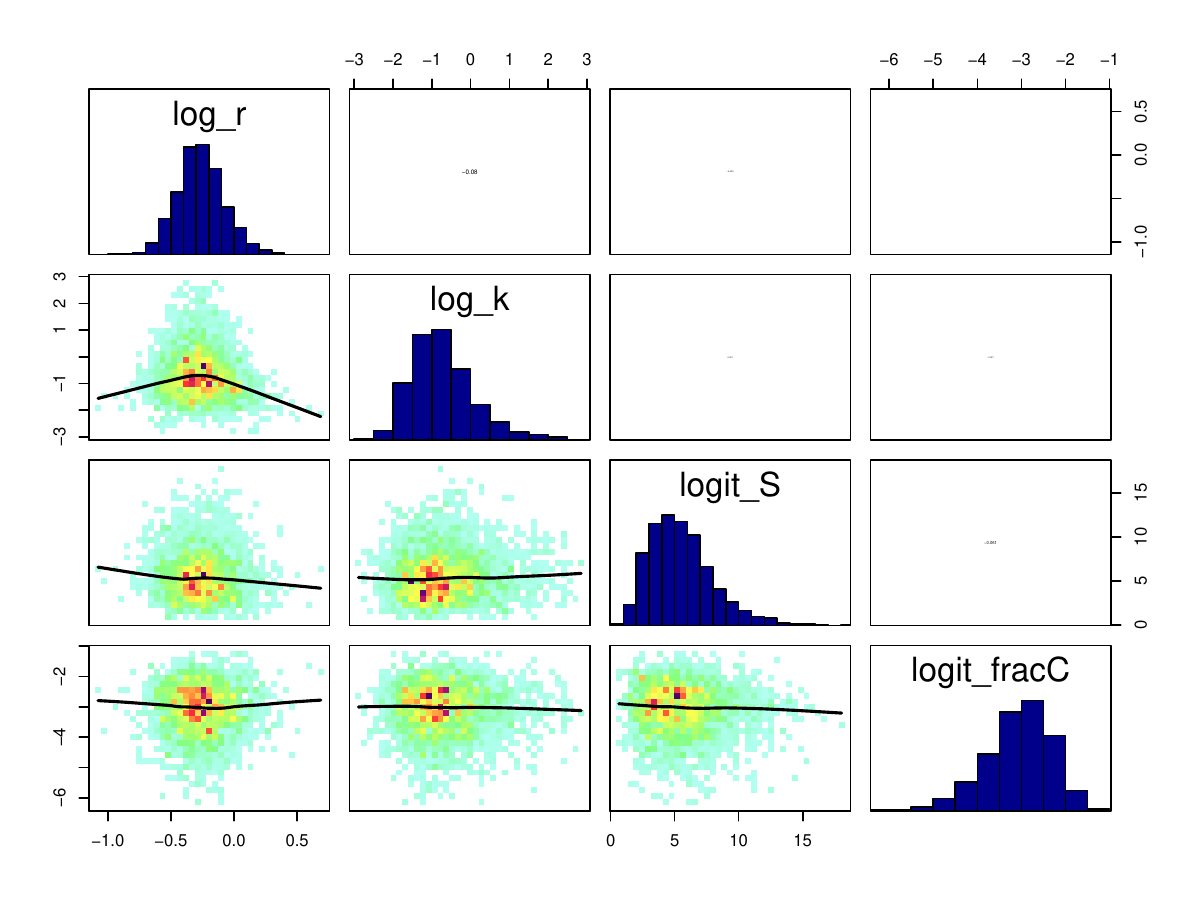}
    
    \caption{Posterior correlation plot for the MERS parameters. }
    \label{fig:corr_plot}
\end{figure}

\FloatBarrier

\end{appendices}

\bibliography{references}

\end{document}